\documentclass[smallextended]{article}       

\usepackage{graphicx}
\usepackage[T1]{fontenc}
\usepackage[utf8x]{luainputenc}
\usepackage{amsmath}
\usepackage{amsfonts}
\usepackage[usenames,dvipsnames,svgnames,table]{xcolor}
\usepackage{amssymb}
\PassOptionsToPackage{normalem}{ulem}
\usepackage{ulem}
\usepackage{appendix}
\usepackage{tikz}
\usetikzlibrary{decorations.pathmorphing}
\usetikzlibrary{shapes,arrows,positioning} 
\tikzstyle{vertex}=[circle,fill=blue!25,minimum size=15pt,inner sep=0pt]
\tikzstyle{vertex_clean}=[circle,fill=black,minimum size=7pt,inner sep=0pt]
\tikzstyle{connection} = [draw=black!80]
\tikzstyle{connection1} = [connection,rectangle,fill=red!25]
\tikzstyle{connection2} = [connection,regular polygon,regular polygon sides=5,fill=yellow!25]
\tikzstyle{connection3} = [connection,regular polygon,regular polygon sides=7,fill=green!25]
\tikzset{
	>=stealth',
	arc/.style={
		->,
		thick,
		shorten <=0.2pt,
		shorten >=0.2pt,}
}

\usepackage{placeins}
\usepackage{sidecap}
\usepackage{wrapfig}
\usepackage{caption}
\usepackage{subcaption}
\usepackage{hyperref}
\hypersetup{
	colorlinks=true,
	linktoc=section,
	linkcolor=NavyBlue,
	citecolor=PineGreen,
	unicode=true
}
\makeatletter
\def\cl@chapter{}
\makeatother
\usepackage{cleveref}

\newtheorem{proposition}{Proposition}
\newtheorem{lemma}{Lemma}
\newtheorem{definition}{Definition}
\newtheorem{proof}{Proof}
\newtheorem{theorem}{Theorem}

\newcommand{\CommaPunct}{\mathpunct{\raisebox{0.1ex}{,}}}
\newcommand{\spacearg}[1]{\langle #1_{space}}
\newcommand{\timearg}[1]{#1_{time} \rangle}

\def\spacetime#1#2{$\spacearg{#1}$~$\CommaPunct$ $\timearg{#2}$}
\def\spacetimerange#1#2#3{\spacetime{#1}{#2}$|_{\forall #3}$}
\def\lengthstretch#1{($#1$)-stretch}

\newcommand{\lglg}[1]{\lg\lg{#1}}
\newcommand{\set}[1]{\left\{ {#1} \right\}}
\newcommand{\eps}{\epsilon}
\renewcommand{\l}{\ell}
\newcommand{\colored}[1]{\textcolor{RawSienna}{#1}}
\newcommand{\directedXrQ}{\colored{\overset{\longrightarrow}{X_r^Q}}}

\RequirePackage{fix-cm}

\begin{document}

\title{Efficient Vertex-Label Distance Oracles for Planar Graphs
	\thanks{This work was partially supported by Israel Science
		Foundation grants 794/13 and 592/17, and by the Israeli ministry
                of absorption.}
\thanks{An extended abstract of this work was presented in the 13th
  Workshop on Approximation and Online Algorithms (WAOA 2015), held in
Patras, Greece, 17-18 September 2015.}}

\author{Shay Mozes \and Eyal E. Skop}


\maketitle
\begin{abstract}

We consider distance queries in vertex-labeled planar graphs. For any
fixed $0 < \eps \leq 1/2$ we show how to preprocess a directed planar graph
with vertex labels and arc lengths into a data structure that answers queries of the following form. Given a
vertex $u$ and a label $\lambda$ return a $(1+\eps)$-approximation
of the distance from $u$ to its closest vertex with label $\lambda$.
For a directed planar graph with $n$ vertices, such that the ratio of the largest to smallest arc length is bounded
by $N$,  the preprocessing time is
$O(\eps^{-2}n\lg^{3}{n}\lg(nN))$, the data structure size is
$O(\eps^{-1}n\lg{n}\lg(nN))$, and the query time is
$O(\lglg{n}\lglg(nN) + \eps^{-1})$. 
We also point out that a vertex label distance oracle for undirected
planar graphs suggested in an earlier version of this paper is incorrect.

\end{abstract}

\section{Introduction}

Imagine you are driving your car and suddenly see you are about to run
out of 
gas. What should you do?
Obviously, you should find the closest gas station.
This is the {\em vertex-label distance query problem}. 
Various software applications like Waze and Google Maps attempt to provide such a
functionality. The idea is to preprocess the locations of service
providers, such as gas
stations, hospitals, pubs and metro stations in advance, so that when a user, whose location is not known a priori, asks for the
distance to the closest service provider, the information can be retrieved
as quickly as possible. 
A dual situation is, for example, when a taxi company wants to dispatch a taxi 
from the station closest to the location where the taxi is required.
Clearly, this problem can be solved using a {\em vertex-label}
distance oracle by transposing the original graph.

We study this problem from a theoretical point of
view. We model the network as a planar graph with labeled vertices
(e.g., a vertex labeled as a gas station). We study distance oracles
for such graphs. A  {\em vertex-label distance oracle} is a data
structure that represents the input graph and can be queried for the
distance between any vertex and the closest vertex with a desired label. 
We consider approximate distance oracles, which, for any given fixed
parameter $\eps>0$, return a distance estimate that is at least the true
distance queried, and at most $(1+\eps)$ times the true distance (this
is known as a $(1+\eps)$-stretch).
One would like an oracle with the following properties:
queries should be answered quickly, the oracle should consume little space, and
the construction of the oracle should take as little time as
possible. 
We use the notation \spacetime{O(S(n))}{O(T(n))} to express the dependency of the space
requirement and query time of a distance oracle on $n$, the number of vertices in the graph.

\paragraph{\bf Our results and approach}


For directed planar graphs we give a \lengthstretch{1+\eps} \spacetime{O(\eps^{-1}n\lg{n}\lg(nN))}{O(\lglg{n}\lglg(nN)+\eps^{-1})}
vertex-label distance oracle whose construction time is
$O(\eps^{-2}n\lg^{3}{n}\lg(nN))$. Throughout the paper, $N$ is the
ratio of the largest to smallest arc length. To the best of our knowledge, no
non-trivial directed vertex-label distance oracles were proposed
prior to the current work.


Consider a vertex-vertex distance oracle for a graph with label set $L$.
If the oracle works for general directed graphs 
then the vertex-label problem can be solved easily;
add a distinct apex $v_\lambda$ for each label $\lambda \in L$, and
connect every  $\lambda$-labeled vertex to $v_\lambda$ with a zero
length arc. Finding the distance from a vertex $u$ to label $\lambda$
is now equivalent to finding the distance between $u$ and $v_\lambda$.
The main difficulty in applying this approach to
oracle for directed planar graphs is that adding
apices breaks planarity. In particular, it affects the
separability of the graph. Thus, the reduction does not work with oracles that depend on planarity
or on the existence of separators.

Our contribution is in realizing and showing that the internal workings
of vertex-vertex distances oracles for planar graphs due to
Thorup~\cite{Thorup} can be extended to support vertex
labels. Achieving this modification is non-trivial since introducing
the apices needs to be done in a manner that guarantees correctness
without compromising efficiency.
Thorup's oracles rely on the existence of fundamental cycle separators in
planar graphs, a property that breaks when apices are added to the
graph. We observe, however, that once the graph is separated, Thorup's oracle does not
depend on planarity. We therefore postpone the addition of the apices
till a later stage in the construction of the distance oracle, when
the graph has already been separated. We show
that, nonetheless, approximate distances from any vertex to any label
in the entire graph can be efficiently approximated.


An earlier version of this paper claimed a simplified and more
efficient vertex label distance oracle for undirected planar
graphs. This claim turns out to be incorrect, as we briefly explain
in Section~\ref{sec:undirected}.

\section{Related Work}
We summarize related work on approximate 
distance oracles.
For general graphs, 
no \lengthstretch{2} approximate vertex-vertex distance
oracles with nearly-linear space consumption are likely to exist~\cite{PatrascuR14}.
For any integer $k\geq 2$, Thorup and Zwick \cite{TZ} presented a \lengthstretch{2k-1} 
\spacetime{O(kn^{1+1/k})}{O(k)} distance oracle for undirected graphs
whose construction time is
 $O(kmn^{1/k})$.
Wulff-Nilsen \cite{WN} achieved the same result with 
preprocessing of $O(kn^{1+\tfrac{c}{k}})$ for a universal constant $c$.
Several more improvements of \cite{TZ} have been obtained for unweighted
or sparse graphs (\cite{Baswana1}, \cite{Baswana2}, \cite{Baswana3}).
The current state of the art is due to Chechik \cite{Chechik15}. She
presented a \lengthstretch{2k-1} 
\spacetime{O(n^{1+1/k})}{O(1)} distance oracle 
construction algorithm for undirected graphs
with $O(n^2+m\sqrt{n})$ construction time.

In contrast, vertex-vertex oracles for planar graphs with stretch less
then $2$ have been constructed.
\sloppy
Thorup \cite{Thorup} gave a \spacetime{O(\eps^{-1}n\lg{n}\lg(nN))}{O(\lglg(nN)+~\eps^{-1})} 
\lengthstretch{1+\eps} distance oracle for directed planar graphs, and a \lengthstretch{1+\eps} \spacetime{O(\eps^{-1}n\lg{n})}{O(\eps^{-1})} 
(simplified) distance oracle for undirected planar graphs.

Our result is based on
Thorup's directed oracle, which is described in~\cref{sec:thorup_dist}. 
Klein \cite{Klein02} independently gave a distance oracle for
undirected planar graphs with the same bounds.
Kawarabayashi, Klein and Sommer \cite{KKS} have shown a \spacetime{O(n)}{O(\eps^{-2}\lg^{2}(n))} 
undirected \lengthstretch{1+\eps} distance oracle constructed in $O(n\lg^{2}{n})$ time, inspired by \cite{Thorup}.
They give a trade-off of \spacetimerange{O(\tfrac{\eps^{-1}n\lg{n}}{\sqrt{r}})}{O(r+\sqrt{r}\eps^{-1}\lg{n})}{r \leq n} oracle algorithms.
Sommer et al.~\cite{Sommer} have shown better tradeoffs for oracles for undirected planar graphs.
For the case where $N\in poly(n)$, they achieve  \lengthstretch{1+\eps} \spacetime{O^{*}(n\lg{n})}{O^{*}(\eps^{-1})} oracle, where $O^{*}$ hides $\lg(\eps^{-1})$ and $\lg^*n$ factors.

The vertex-label distance query problem was introduced by Hermelin, Levy, Weimann and Yuster \cite{Hermelin}.
For any integer $k\geq 2$, they gave a \lengthstretch{4k-5} \spacetime{O(kn^{1+1/k})}{O(k)}
vertex-label distance  oracle (expected space) for undirected general
(i.e., non-planar)
graphs. This is not efficient when the number $l$ of distinct labels is $o(n^{1/k})$ (since the trivial solution of storing all pairwise vertex-label distances is better in that case).
They also presented a \lengthstretch{2^{k}-1} \spacetime{O(knl^{1/k})}{O(k)} oracle for undirected graphs.
Chechik \cite{Chechik} improved the latter result.
In the same paper she also presented a \lengthstretch{4k-5} \spacetime{O(knl^{1/k})}{O(k)} (expected space) oracle for undirected graph.

For undirected planar graphs, Li, Ma and Ning \cite{LMN}, building
on~\cite{Klein02}, constructed a \lengthstretch{1~+~\eps} vertex-labeled oracle with 
\spacetime{O(\eps^{-1}n\lg{n})}{O(\eps^{-1}\lg{n}\lg\Delta)}
bounds. Here, $\Delta$ is the (hop) diameter of the graph,
which can be $\theta(n)$. It is also shown in~\cite{LMN} how to avoid
the $\lg\Delta$ factor when $\Delta=O(\lg{n})$.
The construction time of their oracle is $O(\eps^{-1}n\lg^2n)$. 

\L\k{a}cki,  O\'cwieja,
Pilipczuk,  Sankowski, and  Zych~\cite{LackiOPSZ15} developed dynamic
vertex-labeled distance oracles for undirected general and planar
graphs, and used them to maintain approximate solutions for dynamic
Steiner and subgraph TSP problems. They describe a generic scheme
for converting certain distance oracles for undirected planar graphs into dynamic vertex-label distance oracle.
Applying their scheme to one of the slower variants of Thorup's
distance oracles, they obtained a \spacetime{O(\eps^{-1}n \lg n \lg(nN))}{O(\eps^{-1}\lg
	n\lg(nN))}  \lengthstretch{1+\eps} undirected vertex-labeled distance
oracle that also supports merging labels. 

\fussy

Another related work is by Abraham, Chechik,
Krauthgamer and Wieder \cite{ACKW15}, who considered approximate nearest neighbor
search in planar graph metrics. This is the special case of
vertex-labeled distance oracle with only one label. 
For this easier problem they obtained a data structure whose size is
nearly linear in the number of labeled vertices. However, they assume
an {\em exact} vertex-vertex distance oracle is provided.

To the best of our knowledge, no
non-trivial directed vertex-label distance oracles were proposed
prior to the current work.

\section{Preliminaries} 
In this paper we only deal with connected graphs. This is without loss of generality since each weakly connected component can be handled separately. 
For a graph $G$, we denote by $V(G)$ and $E(G)$ the set of vertices and arcs of $G$, respectively. 
Throughout the paper, all graphs are {\em directed} unless stated otherwise. 
We write $e=uv$ to denote an arc from vertex $u$ to vertex $v$. 
We use the term edges when dealing with undirected graphs, 
or when we wish to ignore the orientation of arcs in directed graphs.

Let $\delta: E(G)
\rightarrow \mathbb{R}^+$ be a function assigning lengths to the arcs. We assume that the lengths are normalized so
that the smallest arc length is 1.
The {\em length} of a path
is the sum of lengths of its arcs.
For $u,v \in V(G)$, the {\em distance} between $u$ and $v$ in $G$, denoted $\delta_{G}(u,v)$, 
is the length of a shortest $u$-to-$v$ path in $G$.
We denote by $N$ the maximum length of an arc in $G$. Thus
$\delta_G(\cdot, \cdot) = O(nN)$.
We assume there are no parallel arcs since it suffices to keep just the arc with minimum length within each set of parallel arcs.

A path or a cycle $P$ is {\em simple} if each vertex is the endpoint of at most two
arcs of $P$.
The concatenation of two paths $P_1$ and $P_2$, where the last vertex of $P_1$
is the first vertex of $P_2$, is denoted $P_1 \circ P_2$.
Two paths $P$ and $Q$ {\em intersect} if $V(P)\cap V(Q) \neq \emptyset$.
For a simple path $Q$ and a vertex set $U \subseteq V(Q)$ with $|U|
\geq 2$, we define $\bar
Q$, the {\em reduction} of $Q$ to $U$ as a path whose vertices are
$U$. Consider the vertices of $U$ in the order in which they are
visited by $Q$. For every two consecutive vertices $u_1,u_2$ of $U$ in this
order, there is an arc $u_1u_2$ in $\bar Q$ whose length is the
length of the $u_1$-to-$u_2$ subpath of $Q$.

Let $L=\set{\lambda_{i}}_{i=1}^{l}$ be a set
of $l$ labels.
A vertex-labeled graph is a graph $G$ equipped with a function
$f:V(G)\rightarrow L$. We define $V_\lambda=\set{v\in V(G)|
  f(v)=\lambda}$ to be the set of vertices with label $\lambda$.
For a vertex-labeled $G$ and $\lambda \in L$, we define
$\delta_{G}(u,\lambda)=\underset{w\in V_\lambda}{\min}\delta_G(u,w)$
to be the distance in $G$ from $u$ to the closest $\lambda$-labeled vertex.

A {\em vertex-label distance oracle} is a data structure that, 
for a specific vertex-labeled graph $G$, given a
vertex $v\in V(G)$ and a label $\lambda \in L$, outputs
(an approximation of)  $\delta_{G}(v,\lambda)$.
We note that this problem is a generalization of the basic distance
oracle problem in which each vertex is given a unique label. 
Constructing an $O(nl)$-space vertex-label distance oracle is trivial,
by precomputing and storing the distance between each vertex and
each possible label. The goal is, therefore, to devise an oracle which
requires $o(nl)$ space, while allowing fast queries.

We assume the reader is familiar with basic definitions and properties of planar graphs, such as the definitions of planar embeddings, and planar duality.
Let $G$ be a planar embedded graph. Each vertex of $G$ is embedded to a point in the plane, and each arc is embedded as a curve in the plane between the images of its endpoints, such that the images of distinct arcs are internally disjoint.
The faces of $G$ are the maximally connected regions of the plane  after removing the image of the vertices and arcs of $G$. Each face is identified with the set of arcs and vertices on its boundary.
We denote by $G^*$ the planar dual of $G$. 
The vertices of $G^*$ are the faces of $G$, 
and the arcs of $G^*$ are in one-to-one correspondence with the arcs of $G$. We can therefore refer to the same arc $e$ in both the primal $G$ and the dual $G^*$.
For a graph $G$, we define the size of the graph as $|G|=|V(G)|+|E(G)|$. 
An immediate consequence of Euler's formula is that in a planar graph $G$ 
where each face of $G$ is of size at least 3 
(i.e. there are no parallel edges and no self loops),  $|E(G)| = O(|V(G)|)$. 
The number of faces is also $O(|E(G)|)$.
Therefore, for planar graph $G$, $|G|=O(|V(G)|)=O(n)$.

An undirected cycle $C$ in a directed graph $G$ 
is a set of arcs in $G$ that, when regarded as undirected
edges, form a cycle.
Let $G$ be a directed planar embedded graph, and let $f_\infty$ be the infinite face of $G$.
Let $C$ be a simple undirected cycle in $G$.
The cycle $C$ partitions the plane into two regions. 
A face $f$ is {\em enclosed} by $C$ if it belongs to the part of the partition that does not contain the infinite face.
A vertex $u \in V(G)$ is enclosed by $C$ if it is incident to a face enclosed by $C$. It is {\em strictly enclosed} if, in addition, $u$ is not a vertex of $C$.
An arc $e\in E(G)$ is strictly enclosed by $C$ if 
both faces incident to $e$ are enclosed by $C$. 
An arc $e\in E(G)$ is enclosed by $C$ if $e\in E(C)$ or if $e$ is strictly enclosed by $C$.
The {\em interior} of $C$ is the subgraph induced on $G$ by the edges enclosed by $C$. 
The {\em exterior} of $C$ is the subgraph induced on $G$ by the edges not strictly enclosed by $C$. 
Note that $C$ belongs to both the interior and exterior of $G$. 
The strict interior and strict exterior of $G$ are, respectively, 
the interior and exterior of $G$ without the vertices and edges of $C$.

The following is a specialized statement of the Jordan curve theorem for planar embedded graphs. 
\begin{proposition}\label{prop:Jordan}
Let $G$ be a planar graph. Let $C$ be an undirected simple cycle in $G$. Any path in $G$ between a vertex in the interior of $C$ and a vertex in the exterior of $C$ contains a vertex of $C$.
\end{proposition}

Let $T$ be a rooted spanning tree of an undirected graph $G$.
For $u\in V(G)$, let $T[u]$ denote the unique root-to-$u$ path in $T$. We call $T[u]$ a branch of $T$.
Let $u_1,u_2 \in V(G)$ be two vertices such that $e=u_1u_2 \in E(G) \setminus E(T)$.
The {\em fundamental cycle} of $e$ (with respect to $T$) is the
undirected cycle composed of $E(T[u_1]), e$,
and $E(T[u_2])$. Note that a fundamental cycle might not be simple since $T[u_1]$ and $T[u_2]$ might have a common prefix. In general, a fundamental cycle can be decomposed into a simple prefix path and a simple cycle. We extend the notions of enclosure defined above for simple cycles to fundamental cycles. 
A face $f$ is enclosed by $C$ if $f$ is enclosed by the simple cycle portion of $C$. The remaining definitions are as in the simple cycle case.

\begin{proposition}[\cite{vonStaudt}] 
For any spanning tree $T$ of $G$, 
the set of edges of $G$ not in $T$ form a spanning tree of $G^*$.
\end{proposition}
For a spanning tree $T$ of $G$, 
we use $T^*$ to denote the spanning tree of $G^*$ 
consisting of the edges not in $T$. 

An undirected fundamental cycle $C$ is a {\em fundamental cycle separator} if each of the  interior and exterior of $C$, consists of at most $3/4$ of the faces of $G$.\footnote{The definitions in the literature differ in the choice of constant ($3/4$ in our case) as well as in the quantity according to which the balance of the  separation is defined (number of faces in our case).}  
Lipton and Tarjan~\cite{LiptonTarjan} show that, 
if $G$ is triangulated 
(every face is adjacent to at most~3 vertices) 
then for any spanning tree $T$ of $G$, 
there is an edge not in $T$ whose fundamental cycle with respect to $T$ 
is a fundamental cycle separator.
Goodrich~\cite{Goodrich} observed that 
such an edge can be found by looking for an 
edge-separator in the dual tree $T^*$.
Thorup~\cite{Thorup} used a different construction that separates a graph into three balanced subgraphs. Our discussion of fundamental cycle separators differs from Thorup's, and follows that of~\cite{KMS13}.
This leads to a simpler structure of the oracle, as we discuss in Section~\ref{sec:thorup_dist}.
 
Let $G$ be a planar graph. For the discussion of separators we ignore the directions of arcs in $G$ and treat it as an undirected graph (the direction of arcs does not affect separation properties of the graph). 
The following basic and simple lemma was proved in~\cite{KMS13}.
\begin{lemma}\cite[Lemma 1]{KMS13}\label{lem:fundamental-cycle-separator} 
Let $G$ be a planar graph with face weights such that no face has weight more than $1/4$ the total weight. 
Let $T$ be a spanning tree of $G$. 
Let $T^*$ be the spanning tree of the planar dual $G^*$ of $G$ consisting of the edges not in $T$.
Assume that $T^*$ has maximum degree~3. One can find in $O(|V(G)|)$ time an edge $\hat e$ of $T^*$ such that 
the total face weight enclosed by the fundamental cycle of $\hat e$ with respect to $T$ is at least $1/4$ the total weight and at most $3/4$ the total weight.
\end{lemma}

Let $G$ be a triangulated planar graph with spanning tree $T$. 
Let $T^*$ be the spanning tree of $G^*$ consisting of the edges not in $T$. 
Since $G$ is triangulated, the maximum degree of $T^*$ is 3.
Let $\hat e$ be the edge in Lemma~\ref{lem:fundamental-cycle-separator}, 
and let $C$ be its fundamental cycle with respect to $T$. 
Separating $G$ with $C$ yields two subgraphs, 
the interior $G_{int}$, and exterior $G_{ext}$ of $C$. 
These subgraphs inherit the embedding from $G$. 
Thus, the cycle $C$ is the infinite face of $G_{int}$, 
and the simple cycle portion of $C$ is a face of $G_{ext}$. See \autoref{fig:fund} for an illustration.

One can obtain from $T$ subtrees that span $G_{ext}$ and $G_{int}$. 
The subgraph of $T$ enclosed by $C$ is a spanning tree $T_{int}$ of $G_{int}$, 
and the subgraph of $T$ not strictly enclosed 
by $C$ is a spanning tree $T_{ext}$ of $G_{ext}$. 
Similarly, Consider the two subtrees, $T^*_{ext}$ and $T^*_{int}$, 
obtained from $T^*$ by deleting the edge $\hat e$, and then adding $\hat e$ back to both subtrees.
The subtree $T^*_{ext}$ that contains the dual vertex corresponding to the infinite face $f_\infty$ of $G$ is a spanning tree of $G_{ext}^*$. The other subtree, $T^*_{int}$, is a spanning tree of $G_{int}^*$. 

We note that even though $G$ is triangulated, $G_{int}$ and $G_{ext}$ are not. For example, the face of $G_{int}$ whose boundary is $C$ might not be a triangle. However, since the spanning trees $T^*_{ext}$ and $T^*_{int}$ of $G^*$ are subtrees of $T^*$, their maximum degree is 3. Hence, one can continue to apply Lemma~\ref{lem:fundamental-cycle-separator} recursively 
even though $G_{int}$ and $G_{ext}$ are not triangulated. We will describe such a decomposition in Section~\ref{sec:thorup_dist}. 

\begin{figure}
\begin{center}
\begin{subfigure}
  {.31\textwidth}{\includegraphics[width=0.95\textwidth]{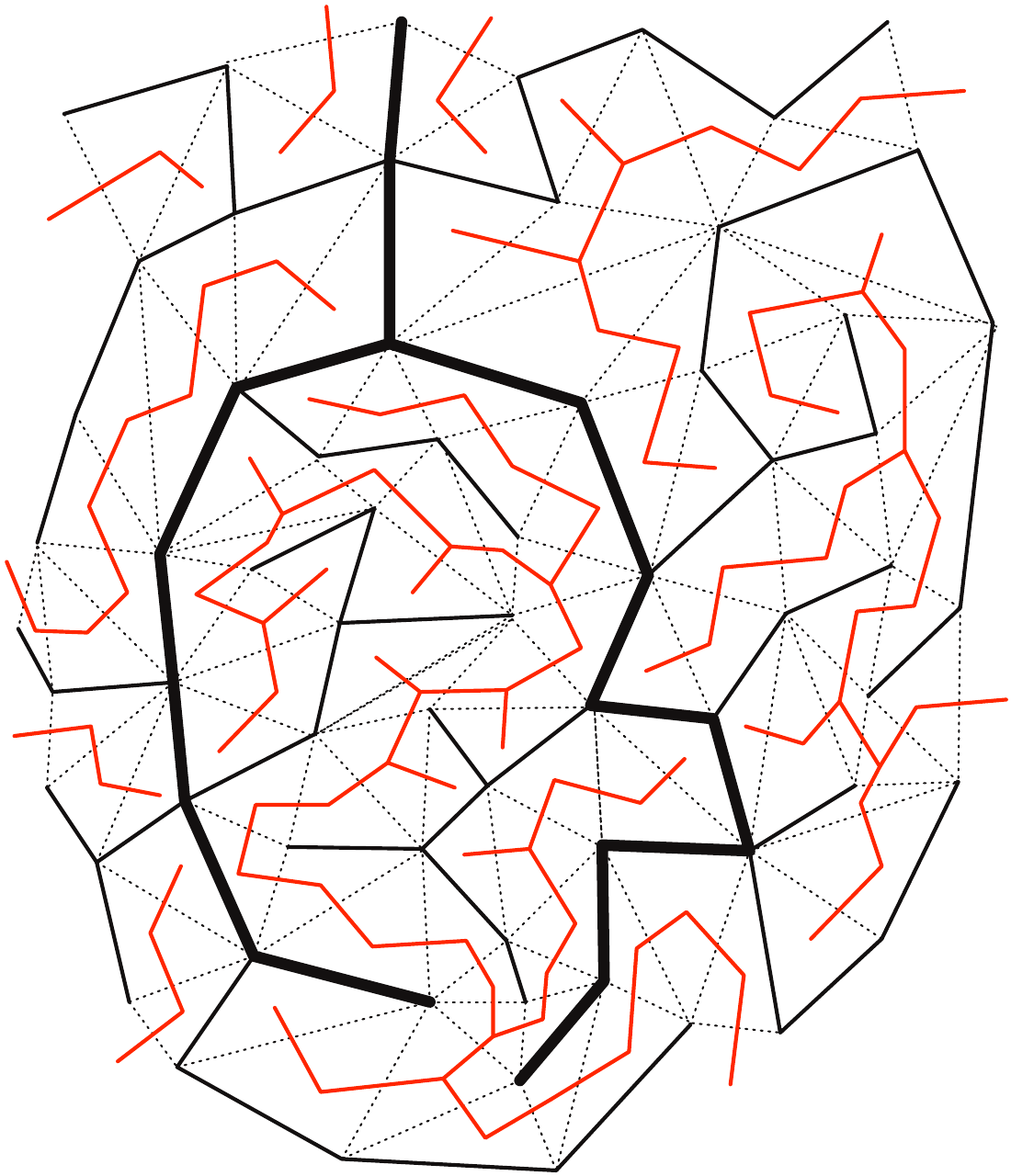}}
\end{subfigure}
\hspace{10pt}
\begin{subfigure}
  {.31\textwidth}\label{fig:fundext}{\includegraphics[width=0.95\textwidth]{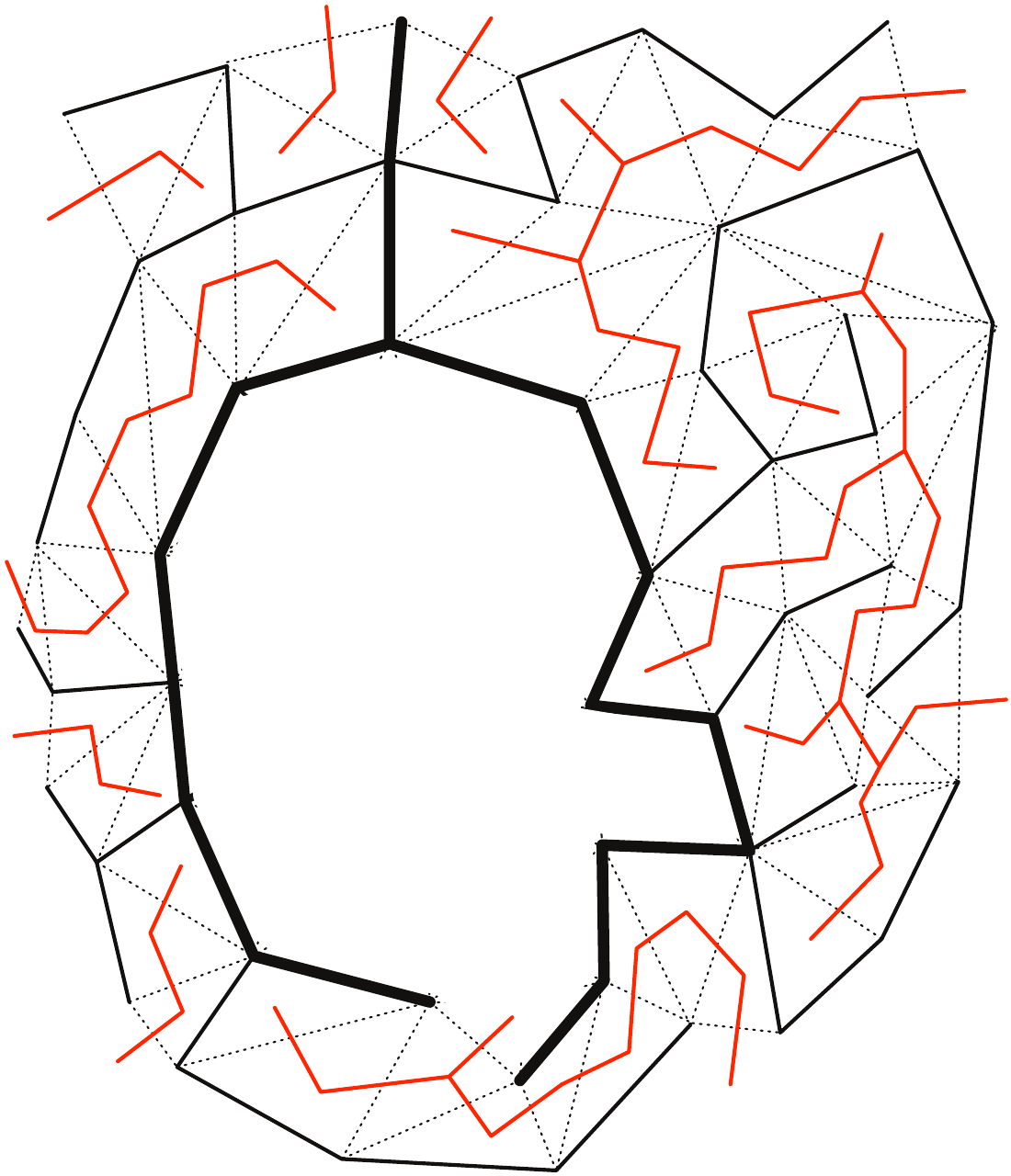}}
\end{subfigure}
\begin{subfigure}
  {.31\textwidth}\label{fig:fundint}{\includegraphics[width=0.95\textwidth]{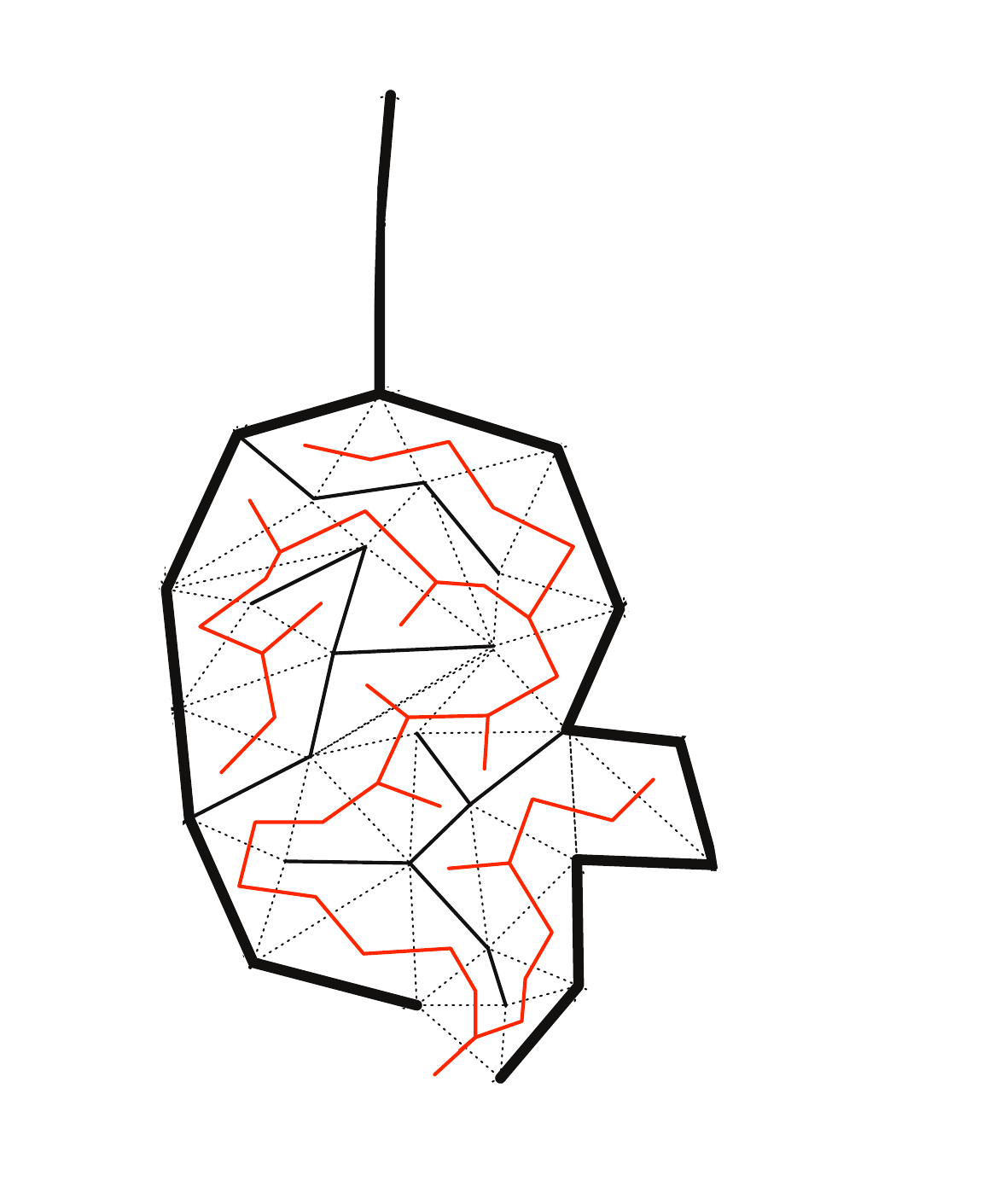}}
\end{subfigure}
\end{center}
\caption{Separating a graph using a fundamental cycle separator. 
	Left: part of a graph $G$ is shown.  Edges of $T$ are black. 
	Two branches of $T$ on a fundamental cycle are solid thick.  
	Edges of the dual tree $T^*$ are shown in red.
	Since only a part of $G$ is shown, this illustration does not show $T^*$ as a tree. Center: part of the subgraph $G_{ext}$ of $G$ is shown. Right: the subgraph $G_{int}$ is shown. Note that the spanning trees $T_{ext}$ and $T_{int}$ are subtrees of $T$, and that the dual spanning trees $T^*_{ext}$ and $T^*_{int}$ are subtrees of $T^*$. Observe that, e.g., the infinite face of $G_{int}$ is not triangulated, yet the maximum degree of $T^*_{int}$ remains 3. } \label{fig:fund}
\end{figure}

\section{Thorup's Approximate Distance Oracle~\cite{Thorup}\label{sec:thorup_dist}}

In this section we outline the distance oracle of
Thorup~\cite{Thorup}. This material is not new, but is necessary
for understanding our results. Our description is somewhat  different from
that of Thorup. It does not  go into all the details of Thorup's oracle. 
Rather, we focus on the aspects 
that our algorithm does not use as black boxes.

Thorup shows that the problem of 
constructing a distance oracle for a
directed graph can be reduced to constructing a
distance oracle for a restricted kind of graphs, 
defined in the following.

\begin{definition}
A set $T$ of arcs in a directed graph $H$ is an $\alpha$-layered spanning tree if it
satisfies the following properties:
\begin{itemize}
\item Regarding the arcs of $T$ as edges, $T$ forms a rooted spanning tree of $V(H)$.
\item Each branch of $T$ can be decomposed into no more than 
$3$ directed shortest paths in $H$, each of 
length at most $\alpha$. These paths may be of opposing directions. 
I.e., they need not all be directed away from the root of $T$.
\end{itemize}
\end{definition}

A graph $H$ is called $\alpha$-layered if it has an $\alpha$-layered spanning tree.
Thorup shows that any graph $G$ can be decomposed into 
$\alpha$-layered graphs of total linear size 
for any $\alpha\in \mathbb{R^+}$. The decomposition is such that any
shortest path of length at most $\alpha$ in G is represented in at
least one of a constant number of $\alpha$-layered graphs in the
decomposition (See \autoref{sec:VLOracle} for
more details). Thus, a natural scaling technique can be used to answer
distance queries in $G$ by answering distance queries in a few
$\alpha$-layered graphs. 
The concept of $\alpha$-layered is important because fundamental cycles of an
$\alpha$-layered spanning tree can be decomposed into a constant
number of directed shortest paths. This property is crucial in the design of
Thorup's oracle.

\begin{definition}
	A scale-$(\alpha,\eps ')$ distance oracle for a graph $H$
        is a data structure that, when queried for $\delta_{H}(v,w)$, returns
	$$d(v,w) \in
	\begin{cases}
	[\delta_{H}(v,w), \delta_{H}(v,w) + \eps '\alpha]  &
        \textrm{if $\delta_{H}(v,w)\leq \alpha$} \\
	[\delta_{H}(v,w), \infty] & \textrm{otherwise} \\
	\end{cases}$$
	
\end{definition}

Thorup shows how to construct a
distance oracle for any graph $G$  using scale-$(\alpha,\eps ')$ distance
oracles for minors of $G$ at several scales.  
This is summarized in the following lemma
(See \autoref{sec:VLOracle} for more details).


\begin{lemma} (\cite[Sections 3.1,3.2,3.3]{Thorup})
	\label{lem:red_scale}
        Let $G$ be a graph. Suppose that,  
        for any $\alpha, \eps' \in \mathbb{R^+}$ and
        any $\alpha$-layered minor $H$ of $G$, one can construct,  in $O(p(|H|,\eps'))$ time, a
        scale-$(\alpha,\eps ')$ distance oracle with space bound
        $O(s(|H|, \eps'))$ and query time $O(t(\eps '))$ (here, $p,s$
        and $t$ are arbitrary functions that only depend on $|H|$ and $\eps'$, not on $\alpha$). 
        Then, one can construct, for any $\eps \in \mathbb{R^+}$,
	a \lengthstretch{1+\eps} \spacetime{O(s(|G|,
          \tfrac{\eps}{4})\lg(|G|N))}{O(t(\tfrac{1}{4})\lglg(|G|N) +
          t(\tfrac{\eps}{4}))} distance oracle for $G$ in $O(p(|G|,\tfrac{\eps}{4})\lg(|G|N))$ time.
\end{lemma}


Planar graphs are closed under taking minors. Thus, by
Lemma~\ref{lem:red_scale}, 
to show a distance oracle for planar graphs, one only needs to show how to construct
scale-$(\alpha,\eps)$ distance oracles for $\alpha$-layered
planar graphs. We next describe Thorup's construction of such oracles.
In Section~\ref{subsec:thorup_preprocess} we explain the
recursive structure. In Section~\ref{subsec:thorup_final_construction}
we describe a non-efficient construction, and in
Section~\ref{subsec:thorup_efficient} we describe how to make the
construction efficient. We note again that all of these constructions are essentially due to Thorup.

\subsection{The recursive decomposition\label{subsec:thorup_preprocess}}

Let $G$ be a directed $\alpha$-layered planar graph 
with an $\alpha$-layered spanning tree $T$. 
Thus, each branch in $T$ can be decomposed into 
at most 3 directed shortest paths. 
We assume that $G$ is triangulated. This is without loss of generality since one can triangulate $G$ with infinite length bidirected arcs. Clearly, this does not affect the shortest paths or the distances in $G$.
For the description of the recursive decomposition we ignore 
the directions of arcs of $G$ and treat it as an undirected graph. 
We stress that this is done only to define the decomposition. 
When describing the oracle we will, of course, 
take the directions of arcs into consideration.

The set of edges not in $T$ forms a spanning tree $T^*$ of $G^*$, and, 
since $G$ is triangulated, the maximum degree of $G^*$ is at most 3. 
We decompose $G$ recursively using the fundamental cycle separator in
\autoref{lem:fundamental-cycle-separator} until each 
subgraph contains a constant number of faces of $G$.
The decomposition can be represented by a binary tree $\mathcal T$ in the following manner. We refer
	to the vertices of $\mathcal T$ as {\em nodes} to distinguish them
	from the vertices of $G$.
	See \autoref{fig:RGD} for an  illustration.

\begin{itemize}
	\item Each node $r$ of $\mathcal T$ is associated with a
	subgraph $G_r$ of $G$. 
	The subgraph associated with the root of $\mathcal T$ is $G$ itself. 
	The spanning tree of $G$ is $T$ and the spanning tree of $G^*$ is $T^*$
	\item Each non-leaf node $r$ of $\mathcal T$ is associated with the fundamental cycle separator $Sep_r$ found by 
	invoking \autoref{lem:fundamental-cycle-separator} on $G_r$. The weight assignment to the faces of $G_r$ used in the invocation of \autoref{lem:fundamental-cycle-separator}  assigns weight $1$ to each face of $G_r$ that is also a face of the original graph $G$ (we call these faces {\em original} faces), and weight $0$ to all other faces of $G_r$ (these faces are called {\em holes}).
	\item Each non-leaf node $r$ has two children $r_0,r_1$.
	The subgraph $G_{r_0}$ associated with the node $r_0$ is the exterior of $Sep_r$ in $G_r$. 
	The subgraph $G_{r_1}$ associated with the node $r_1$ is the interior of $Sep_r$ in $G_r$.

	In both $G_{r_0}$ and $G_{r_1}$ we replace each of the two branches comprising $Sep_r$ 
	by their reduction to their vertices incident to at least one original face.	
	
	As explained in the text following \autoref{lem:fundamental-cycle-separator}, 
	the spanning trees of $G_{r_0}$ and $G_{r_1}$ are subtrees of $T_r$ (with the branches of $Sep_r$ reduced as described above), 
	and the spanning trees of $G^*_{r_0}$ and $G^*_{r_1}$ are subtrees of $T_r^*$. 
	\end{itemize}

\begin{figure}
\begin{center}
\begin{subfigure}
  {.32\textwidth}\label{fig:tri}{\includegraphics[width=0.95\textwidth]{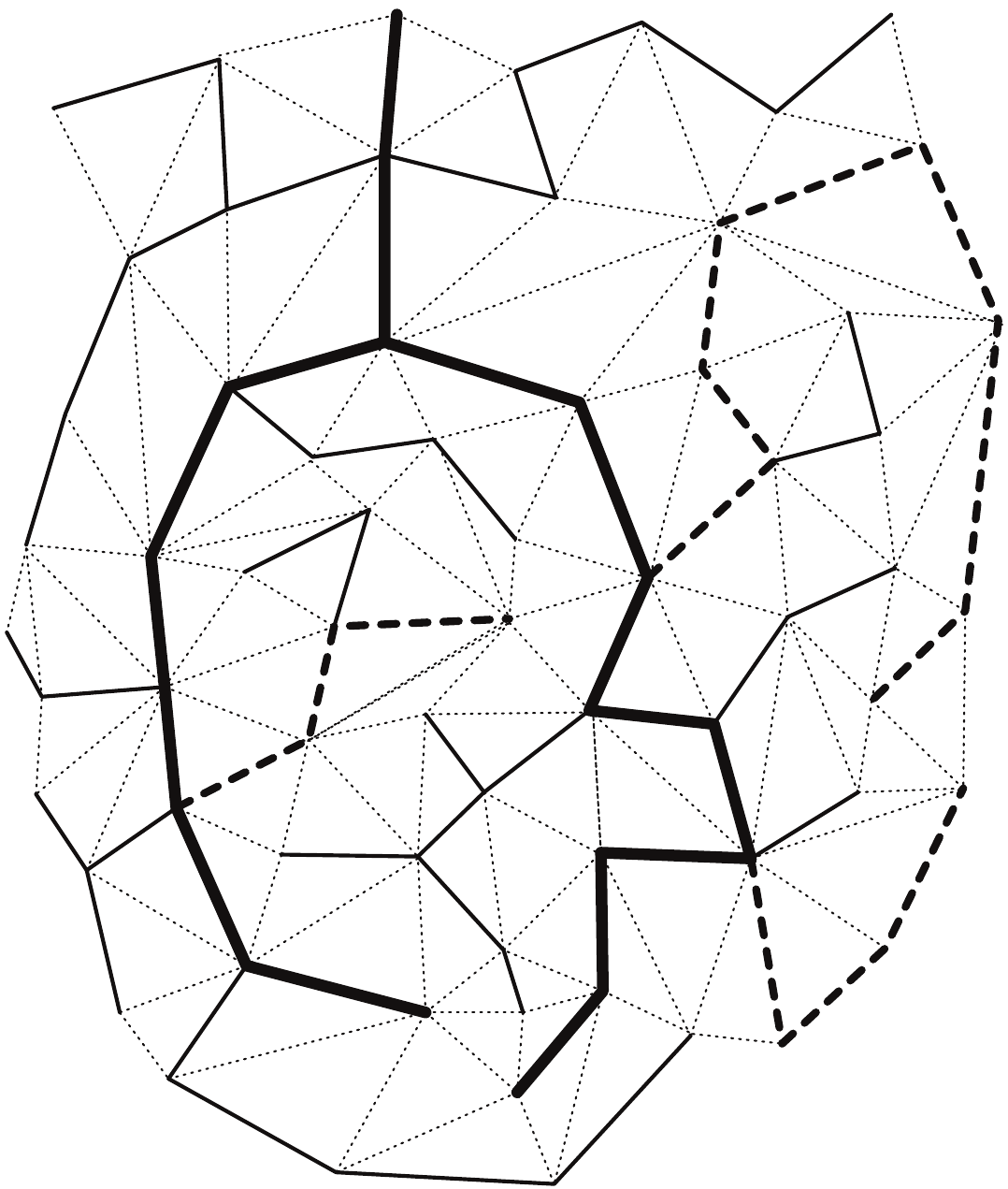}}
\end{subfigure}
\begin{subfigure}
  {.32\textwidth}\label{fig:tri0}{\includegraphics[width=0.95\textwidth]{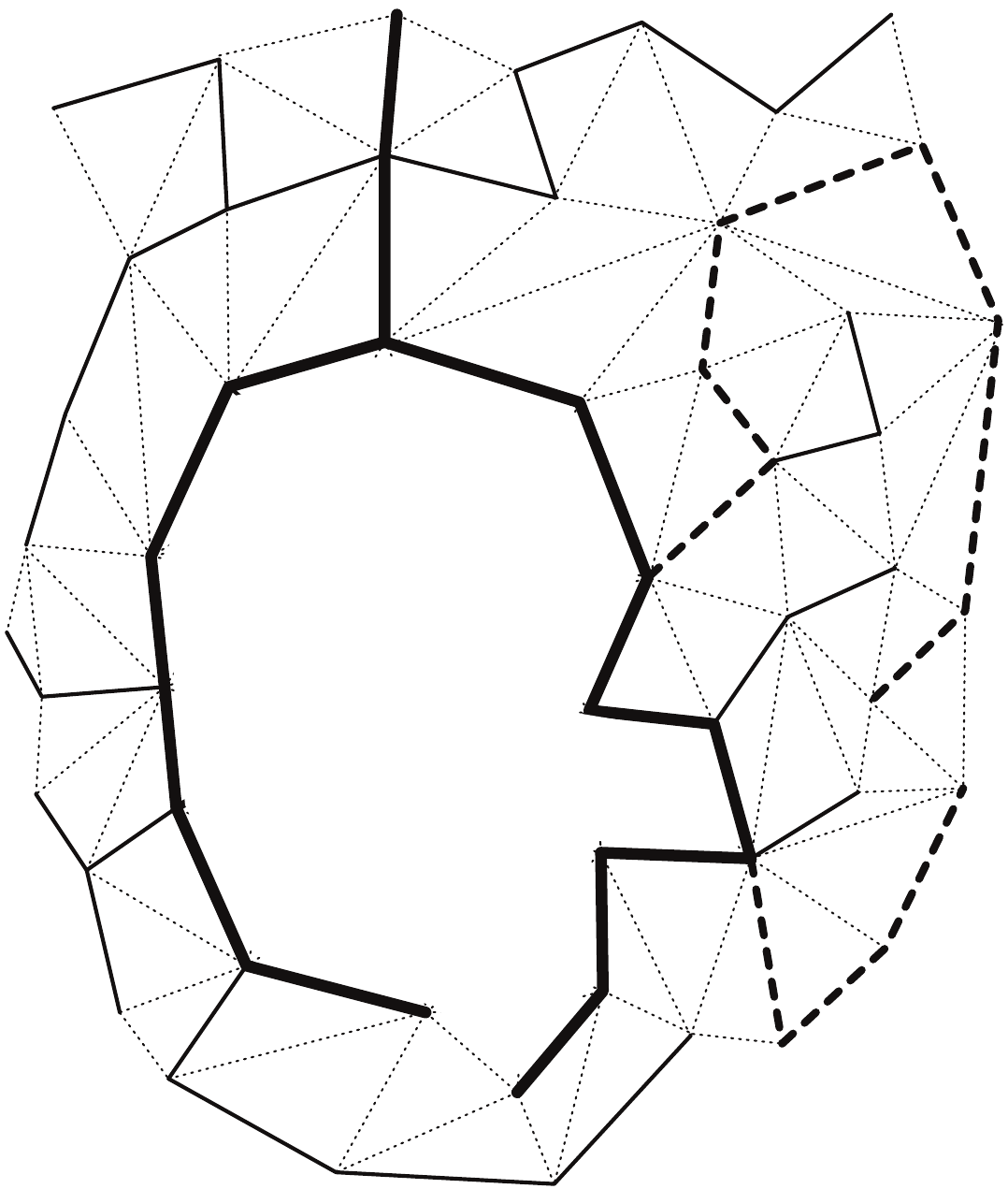}}
\end{subfigure}
\begin{subfigure}
  {.32\textwidth}\label{fig:tri00}{\includegraphics[width=0.95\textwidth]{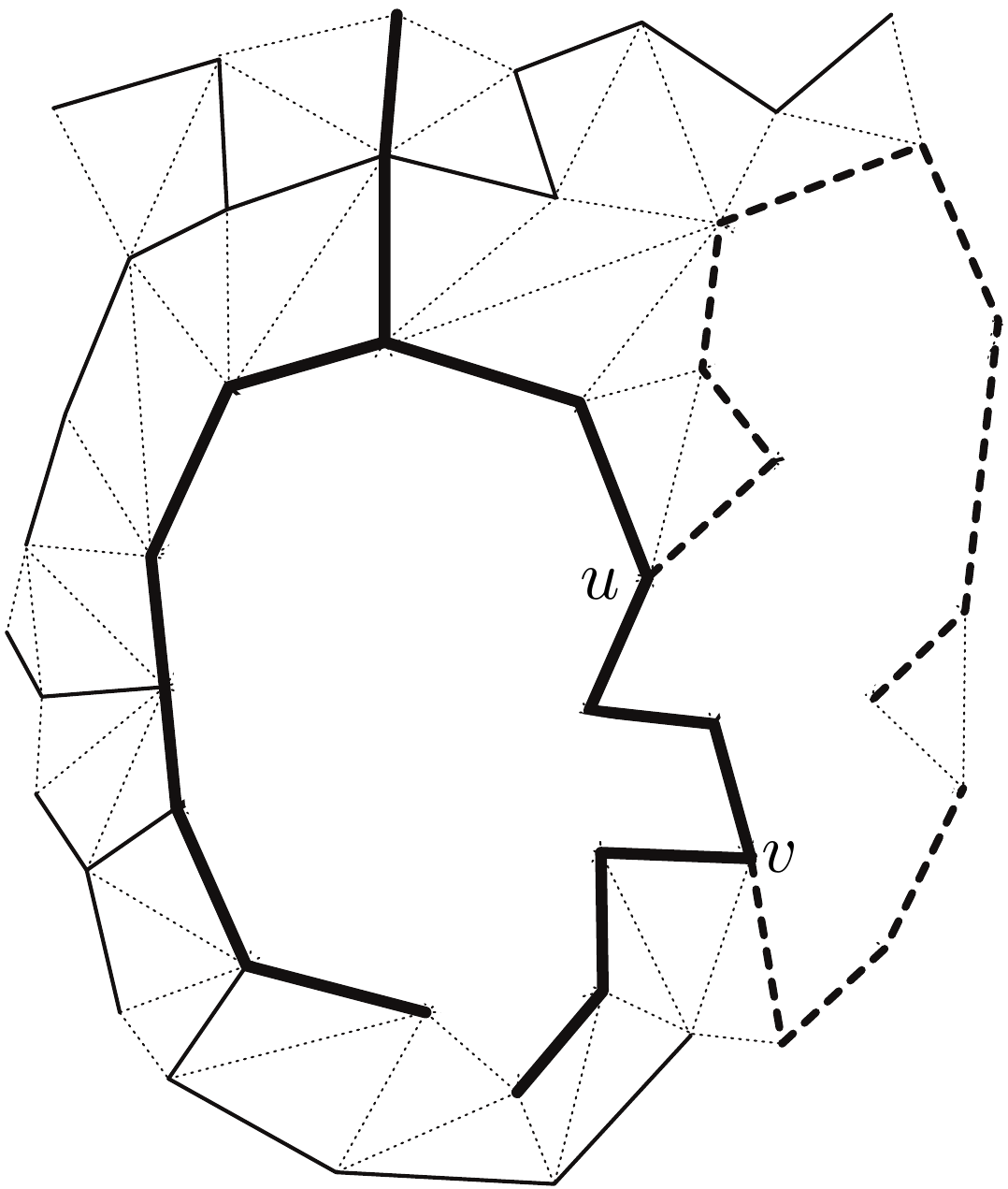}}
\end{subfigure}
\end{center}
\caption{Illustration of the recursive decomposition. 
	Edges of $T$ are black. 
	Edges of the cycle separator at the first level are solid thick, 
	edges of the cycle separators at the second level are dashed. Left: part of the graph $G$ is shown. Center: part of the subgraph $G_0$ of $G$. $G_0$ has a single hole. Right: part of the subgraph $G_{00}$ of $G_0$. $G_{00}$ has two holes. The subpath of $T$ between vertices $u$ and $v$ is reduced into a single edge because internal vertices on this subpath are not incident to original faces of $G$.} \label{fig:RGD}
\end{figure}

Since the number of original faces decreases 
by a constant factor at each level, 
the depth of $\mathcal T$ is $O(\lg n)$.
We will bound the space required to store 
the decomposition tree $\mathcal T$, 
and all the subgraphs $G_r$ generated by the decomposition by $O(|G| \lg|G|)$. This is where reducing the branches of $Sep_r$ is crucial, because the sizes of 
the separators do not decrease along the recursion. We explain this issue in detail.
Recall that the subgraphs created during the recursive decomposition 
have two type of faces; faces that are also original faces of $G$, 
which we call {\em original} faces, 
and faces that are not faces on $G$, which are called {\em holes}. 
The original faces are all triangles, and their number decreases 
by a constant factor at each recursive step. 
The holes are faces that may consist of more than 3 vertices, 
and their number increases by at most one at each recursive step. 
The problem with bounding the size of the subgraphs $G_r$ 
arises from vertices that are adjacent only to holes, 
because the number of such vertices does not necessarily 
decrease along the recursion. Note that such vertices are always vertices of the separator of some ancestor of $r$ in $\mathcal T$ (possibly $r$ itself). 
For this reason we replace the branches of $Sep_r$ 
with their reductions to vertices incident to at least one original face. 
This has the effect of replacing every (possibly long) maximal path 
of edges that are only incident to holes 
with a single reduced edge (that is assigned the length of the path). 
This change reduces the size of the subgraph without changing the distances in the subgraph.
After reducing these paths each vertex in the subgraph is incident to at least one original face. Therefore, the number of vertices in a subgraph 
is at most 3 times the number of original faces in the subgraph.
Namely, for any $r \in \mathcal T$, $|G_r|$ is within a constant multiplicative factor from the number of original faces in $G_r$. Since each original face belongs to exactly one $G_r$ 
at each level of $\mathcal T$, the total size of all subgraphs 
in a single level of the recursion is $O(|G|)$. Thus 
the total size is $O(|G| \lg |G|)$.

\subsection{A scale-$(\alpha,\eps)$ distance oracle\label{subsec:thorup_final_construction}}

The main idea in obtaining an approximate distance oracle is to store just a subset of the pairwise distances in
the graph, from which all approximate distances can be computed
efficiently.
Fix some $\eps > 0$. 
\begin{definition}[{$\eps$-covering connections set}]
Let $H$ be a graph. Let $Q$ be a shortest path in $H$ of
length at most $\alpha$. Let $v$ be a vertex of $H$. A set $C(v,Q) \subseteq V(Q)$ is called an {\em $\eps$-covering connections set} from $v$ to $Q$ with 
{\em connection lengths} $\l:$~$V(H)$~$\times$~$V(Q)$~$\rightarrow$~$\mathbb{R}^{+}$ if, for every $q^* \in V(Q)$, if 
$\delta(v,q^*) \leq \alpha$, then there exists some $q\in C(v,Q)$
s.t. $\delta_H(v,q^*) \leq \l(v,q) + \delta_H(q,q^*) \leq \delta_H(v,q^*) +\eps\alpha$.
\end{definition}
Intuitively one should think of a connection length $\ell(v,q)$ as
the true distance $\delta_H(v,q)$. However, as we explain later on, to
achieve efficient construction, $\ell(v,q)$ is sometimes an approximation of
$\delta_H(v,q)$.
An $\eps$-covering connections set $C(Q,v)$ from $Q$ to $v$ with 
connections lengths $\l:$~$V(Q)$~$\times$~$V(H)$~$\rightarrow$~$\mathbb{R}^{+}$ is defined symmetrically. 
If $\delta(q^*,v)\leq \alpha$,
then there exists some $q\in C(Q,v)$
s.t. $\delta(q^*,v)  \leq \delta(q^*,q) + \l(q,v) \leq \delta(q^*,v) + \eps\alpha$.
We use the term {\em $\eps$-covering connections
  set} (or just connections set) to refer to
both the $v$-to-$Q$ (or $Q$-to-$v$) connections as
well as their corresponding connection lengths.

Thorup showed that, given any graph $H$,
 shortest path $Q$ in $H$ of
length at most $\alpha$, and 
for every vertex
$v \in V(H)$ there exists an $\eps$-covering connections set of size $O(\eps^{-1})$.
The utility of $\eps$-covering connections sets is summarized in the following lemma (\cite[Lemma 3.5]{Thorup}):
\begin{lemma}\label{lem:conn_approx}
	Let $Q$ be a shortest path of length at most $\alpha$ in graph $H$.
	Let $P$ be a shortest $u$-to-$w$ path in $H$ of length at
        most $\alpha$ which intersects $Q$.
	For any $\eps > 0$, let $C(u,Q)$ and $C(Q,w)$ be $\eps$-covering connections 
	sets from $u$ to
        $Q$ and from $Q$ to $w$, respectively. Let $H_Q^{uw}$ be a graph with vertices $u,w$, 
	the vertices and arcs of the reduction of $Q$ to the
        connections of $u$ and of
	$w$, and with $u$-to-$Q$ and $Q$-to-$w$ arcs whose lengths are the corresponding connection lengths of $C(u,Q)$ and $C(Q,w)$.
Then 
		\begin{equation}
		\delta_{H_Q^{uw}}(u,w) \leq \delta_{H}(u,w) + 2\eps\alpha
		\end{equation}
\end{lemma}

A  {\em Lowest Common Ancestor} (LCA) data structure for a tree $T$
is a data structure that, given any two nodes $x,y$ of $T$, returns
the node furthest from the root that is an ancestor of both $x$ and $y$.
Harel and Tarjan~\cite{HarelTarjan} (and many other subsequent
simpler and practical results) show how to construct, in linear time, an LCA data structure of linear size and constant query time.

Let $u,w$ be vertices of $G$. Let $r_u, r_w$ be leaves of
$\mathcal T$ such that $u \in G_{r_u}$ and $w \in G_{r_w}$.
Let $r$ be the LCA of $r_u$ and $r_w$ in $\mathcal T$.
We denote by $S_r$ the set of six directed shortest paths 
in the decomposition of $Sep_r$ (recall that each branch in an $\alpha$-layered tree can be decomposed into at most 3 directed paths, and that $Sep_r$ is composed of two such branches). 
If $r_u\neq r_w$, $u$ and $w$ are separated by $Sep_r$,
and so, by \autoref{prop:Jordan}, every $u$-to-$w$ path in $G$ 
must intersect some path of $S_r$.
Therefore, by \autoref{lem:conn_approx}, to be able to
approximate $\delta_G(u,w)$, it suffices to keep, for every $Q \in S_r$, connections $C(u,Q)$ and $C(Q,u)$, where the connection lengths reflect distances in $G$ (not in $G_r$). To stress that the connection lengths are in $G$, we call such connections {\em global} connections.
If $r=r_u=r_w$, then $r$ is a leaf of $\mathcal T$, so the size of $G_r$ is constant.  All such distances (in G) between pairs of vertices in $G_r$ are stored explicitly by the oracle.

The distance oracle keeps the following items for every internal node $r \in \mathcal T$ and for
every vertex $u \in G_r$:
\begin{enumerate}
	\item global $\eps/2$ connections $C(u,Q)$ for all $Q \in S_r$.
	\item global $\eps/2$ connections $C(Q,u)$ for all $Q \in S_r$.
\end{enumerate}
These connections, over all $u \in G_r$ and all paths in $S_r$ are called the 
connections of $r$.
In addition, the data structure stores:
\begin{enumerate}
\setcounter{enumi}{2}
	\item A mapping of each vertex $v\in V(G)$ to 
          some leaf node $r_v \in \mathcal
	T$ s.t. $v \in G_{r_v}$.
	\item A lowest common ancestor data structure over $\mathcal T$.
	\item $\eps\alpha$-additive approximation of $\delta_G(u,v)$ for all $u, v \in V(G)$ 
	such that $r_u = r_v$.
\end{enumerate}

We next describe how a query is performed.
Given a $u$-to-$w$ distance query, let $r$ be the lowest common ancestor of
$r_u$ and $r_w$ in $\mathcal T$. 
If $r=r_u=r_w$, the query algorithm returns their length approximation (item 5 above) in $O(1)$ time.
Otherwise, the query algorithm computes, 
for each path $Q \in S_r$ the length of a 
shortest $u$-to-$w$ path that intersects
$Q$ using the $\eps/2$-covering sets $C(u,Q)$ and $C(Q,w)$. 
By the construction of $\mathcal T$, 
the number of such paths $Q$ is constant. 
If $\delta(u,v)\leq\alpha$, one can compute the distance
estimate within $\eps\alpha$ additive error (see \autoref{lem:conn_approx}) for each $Q$ in $O(\eps^{-1})$ time (see also \cite[Lemma 3.6]{Thorup}). 
Thus, if $\delta(u,v)\leq\alpha$, an additive $\eps\alpha$ distance approximation is
produced in $O(\eps^{-1})$ time.

We now bound the total space required for the oracle.
Since the height of $\mathcal T$ is $O(\lg n)$, and since each original face belongs to exactly one subgraph at each level of $\mathcal T$,  
each original face belongs to $G_r$ for $O(\lg n)$ nodes $r$ of
$\mathcal T$. For each of the $O(1)$ shortest paths in $S_r$ of each such node $r$,
and each original face $f$ (of $G$) in $G_r$, 
each of the 3 vertices of $f$ has
a set of $O(\eps^{-1})$ connections. This gives a total of $O(\eps^{-1}n\lg{n})$ connections in items 1,2.
Items 3,4 require $O(n)$ space.
Since there are $O(n)$ leaves, each with a subgraph of constant size, storing the additional length approximations in item 5 also requires $O(n)$ space.
Hence the total space required by the oracle is  $O(\eps^{-1}n\lg{n})$.

The simplifications in our presentation 
of Thorup's oracle compared to the original description in~\cite{Thorup}
stem from the fact that all fundamental cycle separators 
used throughout our recursive decomposition are obtained via subtrees 
of a single spanning tree $T^*$ of $G^*$. 
Since $G$ is assumed to be triangulated, $T^*$ and 
all its subtrees have maximum degree 3, 
so there is no need to retriangulate the subgraphs 
$G_r$ along the recursive decomposition. 
This leads to a significant simplification because 
it implies that all cycle separators used in our constructions 
consist of original arcs of $G$. 
Therefore, by \autoref{prop:Jordan}, it suffices to check connections 
only on the separator of the LCA of $r_u$ and $r_v$ in order 
to approximate the distance from $u$ to $v$. 
In Thorup's construction a cycle separator $Sep_r$ 
computed for some $G_r$ may consist of artificial edges 
introduced to triangulate $G_r$. 
In this situation, even though $Sep_r$ separates $u$ from $v$ in $G_r$, 
there may exist in $G$ $u$-to-$v$ paths that do not intersect $Sep_r$. 
Therefore, in order to approximate the distance from $u$ to $v$ 
one needs to check connections on all the separators 
of all the ancestors of the LCA of $r_u$ and $r_v$. 
Thorup deals with this by defining the concept of the {\em frame} of $G_r$. 
The use of frames leads to additional complications since one needs 
to make sure that the number of paths in each frame is constant. 
Our construction avoids the need for frames and its associated complications.

\subsection{Efficient construction}\label{subsec:thorup_efficient}
We now mention some, but not all the details of Thorup's
$O(\eps^{-2}n \lg^3 n)$-time construction algorithm. 
The computation of the connections and connection lengths is done
top-down the decomposition tree $\mathcal T$.
Thorup~\cite[Lemma 3.11, 3.12]{Thorup}
describes a divide-and-conquer procedure
that constructs the connections $C(u,Q)$ for all vertices
$u$ in a graph $H$, and a single shortest path $Q$. A symmetric
procedure computes $C(Q,u)$.
We summarize the procedure in the following lemma.

Let $H$ be a graph. Let $Q$ be a shortest path in $H$. 
Let $sssp(Q,H)$ be a function s.t. for any subgraph $H_{0}$
of $H$, and any vertex $q \in Q_0$, where $Q_{0}$ is the reduction of
$Q$ to $V(H_{0})$, we can compute single source shortest paths
from $q$ in the graph $Q_0 \cup H_0$ in $O(sssp(Q,H)|E(H_0)|)$ time.
It is easy to see that a standard implementation of Dijkstra's
algorithm with priority queue implies $sssp(Q,H) = O(\lg{|E(H)|})$.
If $H$ is planar, then $sssp(Q,H)=O(1)$ by the algorithm of Henzinger et
al. \cite{Henzinger}. 
\begin{lemma}
	\cite[Lemmas 3.11,3.12]{Thorup}
	\label{lem:small_ordered_cover}
	Let $Q$ be a shortest path of length at most $\alpha$ in a
        directed graph $H$. 
        Connections $C(u,Q)$ and $C(Q,u)$ for 
        all vertices $u \in V(H)$, where the connection lengths correspond to distances in $H$, can be computed in 
	$O(\eps^{-1}sssp(Q,H)|E(H)|\lg|V(Q)|)$ total time.
\end{lemma}

Using \autoref{lem:small_ordered_cover} on (planar) $G_r$ 
for all internal $r \in \mathcal T$ and all $Q \in S_r$ takes total 
$O(\eps^{-1}|G|\lg^3|G|)$ time 
(by a similar analysis to the one used to bound the space, see also \autoref{lem:construction_runtime}), but does not generate global connections. Because the lemma is applied only to $G_r$, the connection lengths are with respect to $G_r$, not to $G$. We call  such connections {\em local} connections (of $r$).
For the sake of 
later computations of global $\eps/2$ connections,
the algorithm first computes 
local $\eps/6$ connections for all internal $r \in \mathcal T$ and all $Q \in S_r$.

We next describe how to efficiently compute global connections. We discuss the connections of the form $C(Q,u)$. Computing $C(u,Q)$ is similar.
Recall that global connection lengths reflect distances in the
entire graph, not just in $G_r$. Applying
\autoref{lem:small_ordered_cover} on $G$ for every $r$ would take
quadratic time, which is too much.
Instead, the computation is done top-down $\mathcal T$, augmenting $G_r$ 
with the local connection lengths of ancestors of $r$ in $\mathcal T$, 
which have already been computed, and represent distances outside $G_r$. 
This is done as follows.

\begin{lemma}
	\label{lem:type_1_construction}
	Let $r \in \mathcal{T}$ and $\eps > 0$. 
	Global $\eps/2$ connections 
	for $r$ can be
	computed 
	in $O(\eps^{-2}|V(G_r)|\lg^{2}{|V(G)|})$ time
	using just $G_r$ and 
	local $\eps/6$ connections
	of all ancestors of $r$.
\end{lemma}

\begin{proof}
	Let $Q$ be a path in $S_r$.
	Let $X_r^Q$ be the graph composed of:
	\begin{itemize}
		\item The vertices of $G_r$.
		\item The arcs of $Q$.
		\item For each ancestor $r'$ of $r$ (including $r$ itself), 
		for each path $Q' \in S_{r'}$, 
		let $V(Q')^\circ$ be the vertices of 
		$Q'$ that have (local) connections to $V(G_r)$ or (local) connections from $V(Q)$:
			\begin{itemize}
				\item The vertices and arcs of ${\bar Q}'$, the reduction of
				$Q'$ to vertices of $V(Q')^\circ$.
				\item Arcs representing the
				(local) connection lengths of ${\bar Q}'$ for connections from $V(Q')^\circ$ to $V(G_r$) and from $V(Q)$ to $V(Q')^\circ$. 
			\end{itemize}
		
	\end{itemize}
	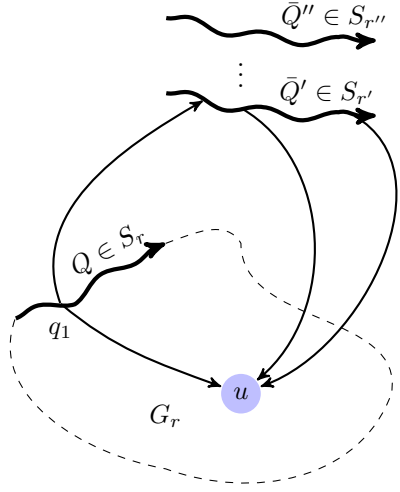
\begin{figure}
		\centering
		\makebox[.5\textwidth]{	
			\begin{tikzpicture}
			
			\node (q1) at (-3,0) {};
			\node (q11) at (-1.5,0.8) {}; 
			\node (q12) at (-2.5,0.2) [xshift=1mm,label=below:$q_1$] {}; 
			\node (q2) at (-1,1) {};
			\node (c1) at (0,1) {}; 
			\node (c2) at (2,-1) {}; 
			\node (c3) at (-1,-2) {}; 
			
			\draw [arc,ultra thick,decorate,decoration = {snake,amplitude =.7mm, segment length = 10mm}] (q1.center)--(q12.center)--(q11.center)--(q2.center) node [draw=none,midway,yshift=0.1mm,xshift=-5.1mm,rotate=30] 
			{$Q \in S_r$};
			\draw [dashed] (q2.center) to [out=20,in=120] (c1.center) to [out=-90,in=80] (c2.center) 
			to [out=-110,in=-20] (c3.center) to [out=170,in=-110] (q1.center);
			
			\node (p1) at (-1,3) {};
			\node (p10) at (-0.5,2.9) {}; 
			\node (p11) at (0,2.8) {}; 
			\node (p12) at (1.5,2.7) {}; 
			\node (p2) at (1.8,2.7) {};
			
			\draw [arc,ultra thick,decorate,decoration = {snake,amplitude =.7mm, segment length = 10mm}] (p1.center) -- (p10.center) -- (p11.center) -- (p12.center) -- (p2.center) node [draw=none,midway,above=0.0cm, xshift=-0.5cm] {\textbf{${\bar Q}'\in S_{r'}$}};

			\node (g1) at (-1,3+1.0) {};
			\node (g10) at (-0.5,2.9+1.0) {}; 
			\node (g11) at (0,2.8+1.0) {}; 
			\node (g12) at (1.5,2.7+1.0) {}; 
			\node (g2) at (1.8,2.7+1.0) {};
			
			\draw [arc,ultra thick,decorate,decoration = {snake,amplitude =.7mm, segment length = 10mm}] (g1.center) -- (g10.center) -- (g11.center) -- (g12.center) -- (g2.center) node [draw=none,midway,above=0.0cm, xshift=-0.4cm] {\textbf{${\bar Q}''\in S_{r''}$}};
			
			\node (dots) [below of=g11, yshift=5.5mm] {$\vdots$};
			
			\node[vertex] (u) at (0,-1) {$u$};
			\node (GrCirc) at (-1,-1.3) {\textbf{$G_r$}};

			\draw [arc,thick] (q12.center) to [out=-40,in=160] (u);

			\draw [arc,thick] (p11.center) to [out=-30,in=40] (u);		
			\draw [arc,thick] (p12.center) to [out=-30,in=20] (u);
			
			\draw [arc,thick] (q12.center) to [out=110,in=-150] (p10.center);		
			
			\end{tikzpicture}
		}
		\caption{
			\label{fig:XrQ}
						The figure illustrates a part of $X_r^Q$.
						The dashed circle represents $Sep_r$. 
						Only solid arcs are part of $X_r^Q$. 
						A shortest path from $q_1 \in Q$ to $u \in V(G_r)$ 
						might be enclosed in $G_r$ (local connections $C(Q, u)$)
						or approximated through a path which intersects a separator of an ancestor of $r$
						(in the figure, it is composed of arcs from $C(q_1,{\bar Q}')$, ${\bar Q}'$ and $C({\bar Q}',u)$).
		}
	\end{figure}
	
	Since each vertex of $G_r$ has $O(\eps^{-1})$ connections to each
	path in the separator of each of $r$'s $O(\lg |V(G)|)$ ancestors, the size
	of $X_r^Q$ is $O(\eps^{-1}|V(G_r)|\lg |V(G)|)$.
	
	The algorithm applies \autoref{lem:small_ordered_cover} 
	on $X_r^Q$ to generate $\eps/6$-covering connection sets 
	from $Q$ in $X_r^Q$.
	We prove below that for any path $Q\in S_r$, and any 
	$q \in V(Q)$ and $u \in V(G_r)$, 
	$\delta_{X_r^Q}(q,u)\leq \delta_{G}(q,u) + \tfrac{\eps}{3}\alpha$. (The fact that $\delta_{G}(q,u) \leq \delta_{X_r^Q}(q,u)$ is obvious.)
	Hence, the $\eps/6$ connections computed in $X_r^Q$ are in fact global $(\eps/3+\eps/6)$ connections in $G$, i.e., global $\eps/2$ connections.  
	
	
Let $q \in V(Q)$, and $u \in V(G_r)$.
	Consider a shortest $q$-to-$u$ path $P$.
	Let $r'$ be the rootmost ancestor of $r$ such that $P$
    intersects $S_{r'}$. I.e., $P$ is confined to $G_{r'}$.
    If $r'=r$ then $P$ is confined to $G_r$, and because 
    there are arcs representing local $\eps/6$ connections lengths for $C(Q,u)$,
    $\delta_{X_r^Q}(q,u)\leq \delta_{G}(q,u) + \tfrac{\eps}{6}\alpha$.
	Otherwise, $r'$ is a strict ancestor of $r$. 
	Let $Q'$ be a path of $S_{r'}$ intersected by $P$. 
	Consider the local $\eps/6$ connection lengths 
	from $q$ to $\bar{Q'}$ in $r'$ and the
        local $\eps/6$ connection lengths from
	$\bar{Q'}$ to $u$. These connection lengths were calculated with respect to exact distances in $G_{r'}$. 
	Since $\bar{Q'}$, $C(q,\bar{Q'})$ and $C(\bar{Q'},u)$ are all present in $X_r^Q$,
	it follows from \autoref{lem:conn_approx}
	that $\delta_{X_r^Q}(q,u)\leq \delta_{G_{r'}}(q,u) + 2\tfrac{\eps}{6}\alpha$.
	Since $P$ is confined to $G_{r'}$, $\delta_{G_{r'}}(q,u)=\delta_G(q,u)$,
	and so $\delta_{X_r^Q}(q,u)\leq \delta_{G}(q,u) + \tfrac{\eps}{3}\alpha$, as claimed.
%

To analyze the running time, first note that $X_r^Q$ can be easily constructed in $O(|X_r^Q|)$ time by storing, when constructing the oracle, for each vertex of $G$ its connections on each of the separators in each of the subgraphs it belongs to. The lengths of the arcs of $X_r^Q$ can be obtained within the same time bound since they are either stored as connection lengths, or correspond to distances between vertices on a shortest path comprising part of a separator. In the latter case a distance can be retrieved in constant time by storing, when constructing the oracle, for each vertex on a shortest path of a separator, its distance from the beginning of the path.

We next show that $sssp(Q, X_r^Q)=O(1)$.
That is, we show how to compute a shortest path tree in $X_r^Q$ rooted at any vertex of $Q$ in $O(|E(X_r^Q)|)$ time. 
	We first go through the vertices of $Q$ in order.
	For each vertex $u \in V(Q)$ we relax all the arcs $uv \in E(X_r^Q)$.
	Let $S$ be the set of all reduced paths of separators of strict ancestors of $r$.
	For each path ${\bar Q}' \in S$ (in any order), we go through the vertices of ${\bar Q}'$ in order. For each vertex $u \in V(\bar Q')$ we relax all the arcs $uv \in E(X_r^Q)$.	
	The computation is correct by the construction of $X_r^Q$; 	
	for any $q \in V(Q)$ and any $v \in V(X_r^Q)$, 
	every $q$-to-$v$ path starts with a subpath of $Q$, 
	followed by a subpath of at most one path ${\bar Q}' \in S$, 
	and then reaching $v$.
	Hence, the relaxation order is correct.

	$X_r^Q$ has $O(\eps^{-1}|V(G_r)|\lg{|V(G)|})$ arcs and vertices 
	because all ancestral separator paths are in reduced form.
	Hence, applying \autoref{lem:small_ordered_cover} to $X_r^Q$ for each path $Q$ of $S_r$ requires $O(\eps^{-2}|V(G_r)|\lg^{2}{|V(G)|})$ time.
	Since there are a constant number of shortest paths in
        $S_r$, this is also the total runtime.
	
\end{proof}

\begin{lemma}
	\label{lem:construction_runtime}
	
	
	Given an $\alpha$-layered graph (along with an $\alpha$-layered spanning tree),
	constructing a scale-$(\alpha,\eps)$ distance oracle takes $O(\eps^{-2}|V(G)|\lg^{3}{|V(G)|})$ time.
	
\end{lemma}

\sloppy
\begin{proof}
    Since a fundamental cycle separator can be found in linear time, constructing the decomposition tree $\mathcal T$ takes $O(|V(G)|\log|V(G)|)$ time. 
	Local connections are computed, for each $r \in \mathcal T$ and each path $Q$ of $S_r$, by applying
	\autoref{lem:small_ordered_cover}
	 to $G_r$ using
	Dijkstra's algorithm.\footnote{Since $G_r$ is planar, one may use
		\cite{Henzinger} instead to achieve $O(\eps^{-1}|E(G_r)|\lg{|V(G)|})$
		complexity. However this is not the bottleneck.} 
	This takes $O(\eps^{-1}|E(G_r)|\lg^{2}{|V(G_r)|})$ time.
	Global connections are computed by applying
	\autoref{lem:type_1_construction} to $X_r^Q$,
	which takes  $O(\eps^{-2}|V(G_r)|\lg^{2}{|V(G)|})$ time by \autoref{lem:type_1_construction}.
	
	Hence, computing all connections for a single $r \in \mathcal T$ takes
	$O(\eps^{-2}|V(G_r)|\lg^{2}{|V(G)|})$. 
	Summing over all $r \in \mathcal T$ in the same depth results in 
	$O(\eps^{-2}|V(G)|\lg^{2}{|V(G)|})$ time.
	Summing over all $O(\lg{|V(G)|})$ depths gives a
	total preprocessing time of $O(\eps^{-2}|V(G)|\lg^{3}{|V(G)|})$.
	
	Finally, computing length approximations for all pair of vertices
	$u, v \in V(G)$ where $r_u = r_v$ (item 5 in the description of the oracle) can also be done within the same
	time bound using a degenerate application of \autoref{lem:type_1_construction}.
	For each leaf node $r$ and for each $u, v \in V(G_r)$, 
	let $r'$ be the subgraph of $r$ consisting of just $u$ and $v$. Define $Sep_{r'}$ to be the singleton vertex $u$, and consider $u$ to be the a trivial local connection from $u$ to $v$  with local connection length $\l(u,v)=\delta_{G_r}(u,v)$. Considering $r'$ as a child of $r$ and applying \autoref{lem:type_1_construction} to $r'$ yields a global connection from $u$ to $v$ in $O(\eps^{-2}\lg^2{|V(G)|})$ time. That is, the length of this connection is an $\eps/2$ approximation of $\delta_G(u,v)$.
	As there are $O(|V(G)|)$ such pairs $(u,v)$ over all leaves of $\mathcal T$, the total preprocessing time remains $O(\eps^{-2}|V(G)|\lg^{3}{|V(G)|})$.
\end{proof}
\fussy

By \autoref{lem:red_scale} and \autoref{lem:construction_runtime}, we
obtain the following theorem:

\sloppy
\begin{theorem}
	\label{thm:thorup_final_result}
		One can construct a \lengthstretch{1+\eps} \spacetime{O(\eps^{-1}n\lg{n}\lg(nN))}{O(~\lglg(nN)~+~\eps^{-1})} distance oracle in $O(\eps^{-2}n\lg^{3}{n}\lg(nN))$ time.
\end{theorem}
\fussy

\section{Scale-$(\alpha,\eps)$ Vertex-Label Distance Oracle\label{sec:vl_sec}}

In this section we show how to adapt Thorup's oracle (\autoref{sec:thorup_dist}) to
the vertex-label case. 
Most of our description details how a scale-$(\alpha,\eps)$ 
vertex-label distance oracle can be constructed efficiently.
At the end of the section we explain why Thorup's argument showing
that scale-$(\alpha,\eps)$ distance oracles can be used to construct a
general distance oracle (\autoref{lem:red_scale}) applies
in the vertex-label case as well.

Thorup's oracle supports one-to-one (vertex-vertex)
distance queries, whereas here we need one-to-many distance queries.
Given two vertices $u,v$, Thorup's oracle finds the 
LCA of $r_u$ and $r_v$  in
$\mathcal T$, and uses its connections to produce the desired distance
approximation. 
In a one-to-many query, we are given the query vertex $u$, but there is no analogue for $v$.
We do not know which $\lambda$-labeled vertex in $G$ is
closest to $u$. 
The minimal distance approximation from $u$ to a $\lambda$-labeled 
vertex in $G_{r_u}$
can be computed in $O(1)$, using the vertex-vertex distance oracle data structure,
as $|G_{r_u}|=O(1)$.
Therefore, we assume, without loss of generality, 
that a shortest $u$-to-$\lambda$ path intersects some path of $S_{parent(r_u)}$.
More precisely, we assume that a shortest $u$-to-$\lambda$ path 
intersects the separator of the leafmost (i.e., furthest from the root) strict ancestor $r$ of $r_u$
in $\mathcal T$ such that $G_r$ contains
some $\lambda$-labeled vertex. The node $r$ takes the role
of the LCA of $r_u$ and $r_v$ in the vertex-vertex query algorithm. In order to be able to use $r$'s connections in a distance query, one must make
sure that $r$'s connections represent approximate distances to
$\lambda$-labeled vertices in the entire graph, not just in $G_r$.


We define a set $\mathcal L$ of new (artificial) vertices, one per label.
For every $r\in \mathcal T$, 
let $\mathcal{L}_r=\set{\lambda\in \mathcal L | V(G_r) \cap V_\lambda \neq \emptyset}$ be the restriction of
$\mathcal L$ to labels present in $G_r$.
Let $\colored{\hat G_r}$ be the graph with vertex set
$V(\colored{\hat G_r})=V(G_r)\cup \mathcal{L}_r$ whose
arcs are the arcs of $G_r$ along with a zero-length arc from each
$\lambda$-labeled vertex of $G_r$ to the corresponding vertex
in $\mathcal L_r$. Note that the number of vertices and arcs in
$\colored{\hat G_r}$ is within a constant factor of those of $G_r$.

In addition to the information stored in the vertex-vertex case (Lemma~\ref{lem:construction_runtime}), 
for every $r~\in~\mathcal T$ and $\lambda~\in~\mathcal L_r$, the
oracle stores connections w.r.t. $\lambda$.
For every shortest path $Q \in S_r$ the oracle stores global $\eps/2$ connections $C(Q,\lambda)$.
Before explaining how to compute these connections we discuss how a distance query
is performed.

Obtaining the distance from $u$ to $\lambda$ is done by
finding the leafmost strict ancestor $r$ of $r_u$ with $\lambda \in \mathcal
L_r$.
The algorithm estimates, for each 
$Q \in S_r$, the length of a shortest $u$-to-$\lambda$ path that intersects
$Q$, using the connections $C(u,Q)$ and $C(Q,\lambda)$ stored for $r$ (Since $\lambda\in \mathcal{L}_r$, $r$ does store $Q$-to-$\lambda$ connections).

Finding $r$ can be done by binary search on the path from $r_u$ to the root
of $\mathcal T$. The number of steps of the binary search is
$O(\lglg n)$. Finding whether a node $r'$ has a vertex with label
$\lambda$ can be done, e.g., by storing all unique labels in $G_{r'}$ in a deterministic dictionary~\cite{Hagerup,Ruzic08}. This dictionary works in the word-RAM model. It stores $k$ elements in $O(k)$ space, requires $O(k \log k)$ construction time, and answers whether an element $x$ is in the dictionary in constant time. Over all nodes of $\mathcal T$ this takes $O(n \log n)$ space, and $O(n \log n \log |L|) = O(n \log^2 n)$ preprocessing time, which are dominated by the overall space and construction time of the oracle. The constant query time of the deterministic dictionary allows us to perform the binary search for $r'$ in $O(\lglg n)$ time. Thus, the query time of our oracle is $O(\lglg n + \eps^{-1})$.

\paragraph{The Construction Algorithm}

It remains to show how the connections are computed.
We begin with the local connections.
For every $r\in \mathcal T$, for every $Q \in S_r$, 
the algorithm computes $\eps/6$ connections sets on
$Q$ w.r.t. each vertex of $\colored{\hat G_r}$
by invoking \autoref{lem:small_ordered_cover} on $\colored{\hat G_r}$. 

As for the quality of approximation, we must show that the
connection lengths to the artificial vertices are useful for approximate distance queries.
For the local connections, the approximation is
immediate because the desired distances are in ${\colored{\hat G_{r}}}$.

We now show how to compute the global connections without invoking
\autoref{lem:small_ordered_cover} on the entire input
graph $G$ at every call.
The crucial point is that the
connections to the artificial vertex $\lambda$ from separators of
ancestors of $r$ represent distances to vertices with label $\lambda$
that are not necessarily in $G_r$.

\begin{lemma}
\label{lem:approx_color_dist_correctness}
Let $r \in \mathcal{T}$. Global $\eps/2$ connections of $r$ to label $\lambda \in \mathcal L_r$ can be
computed using just the local $\eps/6$ connections
of ancestors of $r$.
Computing all global connections to $\mathcal L_r$ for all 
$r \in \mathcal T$ can be done in $O(\eps^{-2}|V(G)|\lg^3{|V(G)|})$ time.
\end{lemma}

\begin{proof}
Let $Q$ be a shortest path of $S_r$.
Let $\directedXrQ$ be the graph composed of the following: (see
\autoref{fig:XrQ_directed} for an illustration)
\begin{itemize}
	\item The vertices $\mathcal L_r$.
	\item The arcs of $Q$.
	\item For each ancestor $r'$ of $r$ (including $r$), 
	for each path $Q' \in S_{r'}$, let $V(Q')^\circ$ be the vertices of 
	$Q'$ that have local connections to $\mathcal L_r$ or 
	local connections from $V(Q)$:
	\begin{itemize}
		\item The vertices and arcs of ${\bar Q}'$, the reduction of
		$Q'$ to vertices of $V(Q')^\circ$.
		\item Arcs representing the
                  local connection lengths of ${\bar Q}'$ from $V(Q')^\circ$ to $\mathcal L_r$ and from $V(Q)$ to $V(Q')^\circ$.
                  
    \end{itemize}
\end{itemize}

\begin{figure}
		
		\centering
		\begin{tikzpicture}
		
		\node (q1) at (-3,0) {};
		\node (q11) at (-1.5,0.8) {}; 
		\node (q12) at (-2.5,0.2) [xshift=1mm,label=below:$q_1$] {}; 
		\node (q2) at (-1,1) {};
		\node (c1) at (0,1) {}; 
		\node (c2) at (2,-1) {}; 
		\node (c3) at (-1,-2) {}; 
		
		\draw [ultra thick,decorate,decoration = {snake,amplitude =.7mm, segment length = 10mm}] (q1.center)--(q12.center)--(q11.center)--(q2.center) node [draw=none,midway,yshift=1.1mm,xshift=-4.1mm,rotate=30]  
		{\textbf{$Q \in S_r$}};
		\draw [dashed] (q2.center) to [out=20,in=120] (c1.center) to [out=-90,in=80] (c2.center) to [out=-110,in=-20] (c3.center) to [out=170,in=-110] (q1.center);
		
		\node (p1) at (-1,3) {};
		\node (p10) at (-0.5,2.9) {}; 
		\node (p11) at (0,2.8) {}; 
		\node (p12) at (1.5,2.7) {}; 
		\node (p2) at (1.8,2.7) {};
		
		\draw [ultra thick,decorate,decoration = {snake,amplitude =.7mm, segment length = 10mm}] (p1.center) -- (p10.center) -- (p11.center) -- (p12.center) -- (p2.center) node [draw=none,midway,above=0.2cm,xshift=-5mm, yshift=-0.2cm] {\textbf{${\bar Q}'\in S_{r'}$}};
		
		\node (g1) at (-1,3+1.0) {};
		\node (g10) at (-0.5,2.9+1.0) {}; 
		\node (g11) at (0,2.8+1.0) {}; 
		\node (g12) at (1.5,2.7+1.0) {}; 
		\node (g2) at (1.8,2.7+1.0) {};
		
		\draw [ultra thick,decorate,decoration = {snake,amplitude =.7mm, segment length = 10mm}] (g1.center) -- (g10.center) -- (g11.center) -- (g12.center) -- (g2.center) node [draw=none,midway,above=0.2cm,xshift=-4mm, yshift=-0.2cm] {\textbf{${\bar Q}''\in S_{r''}$}};
		
		\node (dots) [below of=g11, yshift=5.5mm] {$\vdots$};
		
		\node[vertex] (u1) at (-0.5,-1.5) {$u_1$};
		\node[vertex] (u2) at (1,-1) {$u_2$};
		\node (GrCirc) at (-2,-1.1) {\textbf{$G_r$}};
		\node[vertex] (label) at (3,-2) {\textbf{$\lambda$}};

		\draw [arc,thick] (q12.center) to [out=-40,in=170] (label);

		\draw [arc,thick] (q12.center) to [out=110,in=-150] (p10.center);		
		
		\draw [arc,thick] (q12.center) to [out=110,in=-150] (g11.center);		
		
		\draw [arc,thick] (g12.center) to [out=-20,in=60] (label);
		
		\draw [arc,thick] (p12.center) to [out=-50,in=80] (label);
		\draw [arc,thick] (p10.center) to [out=-70,in=98] (label);

		\end{tikzpicture}
	\caption{
		\label{fig:XrQ_directed}
			The figure illustrates a part of $X_r^Q$
			for the labels case, similarly to \autoref{fig:XrQ}.
			The vertices $u_1$ and $u_2$ are $\lambda$-labeled
			vertices of $G_r$, and are not part of $X_r^Q$.
			Paths from $Q$ to $\lambda$-labeled
			vertices such as $u_1$ and $u_2$ confined to $G_r$ are
			represented in $X_r^Q$ by arcs between $Q$ and $\lambda$. 
			These arcs correspond to the local
			$\eps/6$ connections of $\lambda$ on $Q$ in ${\colored{\hat G_{r}}}$.
			All solid arcs are part of $X_r^Q$. 
			A shortest path from $q_1 \in V(Q)$ to $\lambda \in \mathcal{L}_r$ is
			approximated by connections from $q_1$ to a separator 
			of an ancestor of $r$ and from there to $\lambda$.
			Note that $C(Q',\lambda)$ represent distances from $Q'$ to 
			$\lambda$-labeled vertices that are not necessarily
			in $G_r$.
	}
	
\end{figure}
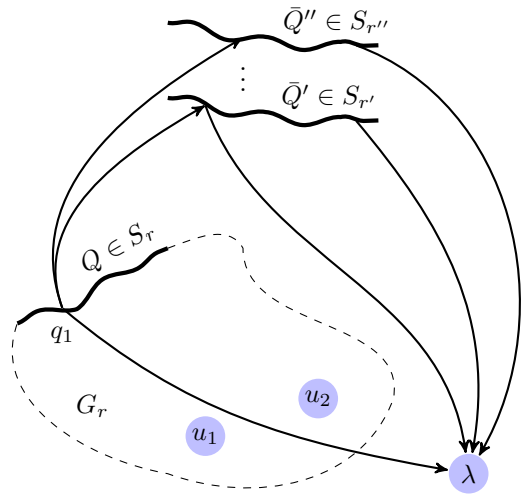

Let $q \in V(Q)$ and $\lambda \in \mathcal{L}_r$.
We show that $\directedXrQ$
approximates the distance from $q$ to its closest
$\lambda$-labeled vertex in $G$ with an additive error of $\tfrac{\eps}{3}\alpha$, namely:
\begin{equation}
\label{eq:construction_approx_colored}
\delta_{\directedXrQ}(q,\lambda)\leq \delta_{G}(q,\lambda) + \tfrac{\eps}{3}\alpha
\end{equation}

Therefore, applying \autoref{lem:small_ordered_cover} to $\directedXrQ$ with $\eps/6$
results in global $\eps/2$ connections, just as in \autoref{lem:type_1_construction}.

Let $u_\lambda$ be a closest $\lambda$-labeled vertex to $q$ in $G$.
Consider a shortest $q$-to-$u_\lambda$ path $P$ in $G$.
Let $r'$ be the rootmost ancestor of $r$ such that $P$ intersects $S_{r'}$.
As in the vertex-vertex case, if $r'=r$ then we are done, 
as the local connection sets of $Q$ represent approximations of distances
from vertices of $Q$ to vertices $\mathcal{L}_r$ in $G$.
Otherwise, $r' \neq r$ and
let $\tilde Q$ be a 
shortest path in $S_{r'}$ that is
intersected by $P$.
Consider the local $\eps/6$ connection lengths from $q$ to $\tilde Q$ and from $\tilde Q$ to $\lambda$ in $r'$.
These lengths were calculated with respect to exact distances in $\colored{\hat G_{r'}}$.
Since $\tilde Q$, $C(q,\tilde{Q})$ and $C(\tilde{Q},\lambda)$ are all present in $\directedXrQ$,
it follows that
$\delta_{\directedXrQ}(q,\lambda) \underset{\autoref{lem:conn_approx}}{\leq} \delta_{\colored{\hat G_{r'}}}(q,\lambda) + 2\tfrac{\eps}{6}\alpha$.
Since $P$ is confined to $G_{r'}$, $\delta_{\colored{\hat G_{r'}}}(q,\lambda)=\delta_{G_{r'}}(q,u_\lambda)=\delta_G(q,\lambda)$,
and the lemma follows.

\end{proof}

The construction algorithm differs from the one in
Section~\ref{subsec:thorup_final_construction} in the existence of the
artificial vertices. Since $|\mathcal L_r| = O(|V(G_r)|)$ 
for any $r \in \mathcal T$, the
total running time and space requirements remain as in Section~\ref{subsec:thorup_final_construction}.

\subsection{Vertex-Label Distance Oracle}\label{sec:VLOracle}
We now explain why \autoref{lem:red_scale} applies
also to the vertex-label case, 
using our construction algorithm for scale-$(\alpha,\eps)$
vertex-label distance oracles.
For this we need to elaborate a bit more on the proof of
\autoref{lem:red_scale}.

The proof of \autoref{lem:red_scale} relies on two reductions~\cite[Lemma 3.2, 3.9]{Thorup}. The
first shows that for any graph $G$, and for any $\alpha >0$, one can construct in linear time a family of
$\alpha$-layered graphs $\{G_i^\alpha\}_i$ such that \begin{enumerate}
\item $\sum |G^{\alpha}_{i}| = O(|G|)$.
\item Each $v\in V(G)$ has an 
index $\iota(v)$ s.t. any $w\in V(G)$ has
$d=\delta_{G}(v,w)\leq \alpha$ iff $d= \min\{
\delta_{G^{\alpha}_{\iota(v)-2}}(v,w) ,
\delta_{G^{\alpha}_{\iota(v)-1}}(v,w), \delta_{G^{\alpha}_{\iota(v)}}(v,w)\}$
\item Each $G^{\alpha}_{i}$ is a minor of $G$. I.e.,  it can be
  obtained from $G$  by contraction and deletion of arcs and
  vertices. In particular, if  $G$ is planar, so is  $G^{\alpha}_{i}$.
\end{enumerate}
Item (2.) means that any shortest path of length at most
$\alpha$ in $G$ is represented in at least one of three fixed graphs
$G_i^\alpha$. Thus,  one
can use scale-($\alpha,\eps$) distance oracles for the
$\alpha$-layered graphs
$\{G_i^\alpha\}$ to implement a scale-($\alpha,\eps$) oracle of $G$. 

The second reduction~\cite[Lemma 3.8]{Thorup} is a scaling argument that shows how to construct
a $(1+\eps)$-stretch distance oracle for $G$ using
scale-$(\alpha,\eps')$ distance oracles
 for $\alpha = 2^{i}$ for all integers $i\in [\lceil\lg(nN)\rceil]$. The
 reduction does not rely on planarity.
Roughly, the idea is that an additive $\frac{\eps}{2} \alpha$-approximation to $\delta(u,v)$ is a multiplicative $(1+\eps)$-approximation if $\frac{\alpha}{2} \leq \delta(u,v)$.
Therefore one can use binary search to identify the appropriate scale from which to get a $(1+\eps)$ approximate distance. See~\cite[Lemma 3.9]{Thorup} for the details.

Now consider the vertex-labeled case.
Let $G'$ be the graph obtained from $G$ by adding apices representing
the labels. All vertices with a specific label are connected in $G'$ to the apex corresponding to this label with zero length arcs. 
A vertex-vertex distance oracle for $G'$ is a vertex-label
distance oracle for $G$, so it suffices to show the former.
Applying Thorup's
second reduction to $G'$, it suffices to show how to construct a
scale-($\alpha,\eps$) distance oracle for $G'$ for any $\alpha,
\eps$. Answering distance queries between vertices of $G'$ that are not apices can be done using a scale-($\alpha,\eps$) vertex-vertex distance oracle for the original planar graph $G$ whose existence was shown by Thorup. Answering distance queries to an apex in $G'$ can be done using a vertex-label scale-($\alpha,\eps$) distance oracle for $G$. It therefore suffices to show a vertex-label scale-($\alpha,\eps$) distance oracle for $G$.
Let $\alpha\in\mathbb{R}^+$.
We use Thorup's first reduction, and construct the vertex-label distance oracle described in the beginning of Section~\ref{sec:vl_sec} for each minor $G_i^\alpha$ of $G$. 
Given $u \in V(G)$ and $\lambda \in \mathcal{L}$
with $\delta_G(u,\lambda)\leq \alpha$,
let $w \in V_\lambda$ be a closest $\lambda$-labeled vertex to $u$.
By the properties of Thorup's first reduction, there is a graph $G_i^\alpha$
in which the $u$-to-$w$ distance is $\delta_G(u,w)$.
Thus, the vertex-label distance oracle for $G_i^\alpha$ will report
a distance of at most $\delta_G(u,\lambda) + \eps\alpha$.

Therefore,  the main theorem follows:
\begin{theorem}
	\label{thm:mod_final_result}
	A \lengthstretch{1+\eps} \spacetime{O(\eps^{-1}n\lg{n}\lg(nN))}{O(\lglg{n}\lglg(nN)+\eps^{-1})}
	vertex-label distance oracle can be constructed
	in $O(\eps^{-2}n\lg^{3}{n}\lg(nN))$ time 
		for a directed planar graph with
	$n$ vertices and maximum arc length $N$.
\end{theorem}

\section{Reporting Approximate Shortest Path}
\label{sec:report}

Thorup describes how to augment his oracle to report  a $u$-to-$v$
path of length $(1+\epsilon)\delta(u,v)$
in time linear in the number of arcs reported.
We provide here a brief description (cf. \cite[Sections 2.7, 2.8, 3.7]{Thorup}).

The algorithm now stores additional information. The $\eps$-cover
construction algorithm (\autoref{lem:small_ordered_cover}) computes shortest path trees
rooted at each connection. In the original description these
trees are discarded once the connection lengths have been recorded. To
report shortest paths, the algorithm stores these trees for all local
connections.
Let $r$ be a node of $\mathcal{T}$. 
For each $v \in G_r$, $Q \in S_r$ and global connection $q \in C(Q,v)$, (i.e., connections computed by invoking the $\eps$-cover
construction algorithm on $X_r^Q$), the algorithm records the rootmost
node $r' \in \mathcal T$ whose separator is intersected 
by the $q$-to-$v$ shortest path in $X_r^Q$. 

The query algorithm for the distance between $u$ and $v$ uses some
global connection. Let $r'$ be the node of $\mathcal T$ 
recorded for that connection. By choice of $r'$ there exists a
$(1+\eps)$-approximate shortest path $P$ between $u$ and $v$ in
$G_{r'}$. The algorithm now uses the local
connections of $r'$ to find $P$ and uses the shortest path trees
stored for those connections to report the edges of $P$.

Storing all shortest path trees does not change the preprocessing time 
but increases the required space to $O(\eps^{-1}n\lg^2{n})$.
This is because each vertex participates in $O(\eps^{-1}\lg{n})$ 
	connections shortest path trees (see \cite[Lemma~3.12]{Thorup}) in $O(\lg{n})$ nodes of $\mathcal{T}$.
Let $\bar{d}$ be the number of arcs of the reported path.
The resulting query time is $O(\eps^{-1} + \lg\lg{n} + \bar{d})$.

We extend this technique to  the vertex label case. 
As in the vertex-vertex case, the query algorithm finds 
a node $r'$ such that there exists in $G_{r'}$ a path $P$ that 
$(1+\eps)$-approximates the shortest path between $u$ and a
$\lambda$-labeled vertex in $G$. It then finds $P$ using the local
connections of $r'$.
The main difference is that at this point the algorithm knows the
distance to the artificial vertex $\lambda$, but not
the identity of the $\lambda$-labeled vertex realizing this distance. However, this
information can be stored along with the shortest path trees for
local connections.
Consider any $r$. Recall that each local connection corresponds to a
shortest path in  $\colored{\hat
  G_r}$. We record, with every local connection to each $\lambda \in \mathcal L_r$,  the identity of the
vertex preceding $\lambda$ in the corresponding shortest path. 

\section{Undirected Vertex-Label Distance
  Oracle\label{sec:undirected}}
The extended abstract of the current paper~\cite{WAOA} contained a
description of a simplified and more efficient version of a
vertex-label distance oracle for undirected planar graphs.
Unfortunately, that oracle is flawed. Specifically, Lemma 6
in~\cite{WAOA} is false. Let $Q$ be a shortest path in an undirected
graph $G$. If one introduces  an artificial vertex $u$ (apex) to $G$
and connects $u$ to vertices of $G$
with zero-length edges, then $Q$ might no longer be a shortest
path. The proof of Lemma 6 in~\cite{WAOA} assumes $Q$ remains a
shortest path, and is therefore incorrect.
This seems to be a fundamental problem with this
approach, which we were not able to correct. This situation is not
problematic in the directed case, because all the arcs that are added
to the graph enter the artificially added vertices. Therefore, adding
these arcs does not change distances between vertices of the graph in
the directed case.

We do mention an improvement to the vertex-label distance oracle
of \cite{LMN} for undirected planar graphs. 
Given a label $\lambda$ and vertex $u$, for any path $Q$, 
when queried for the shortest $u$-to-$\lambda$ path which intersects
$Q$, the algorithm in
\cite{LMN} uses a predecessor search in order to identify a range of $\bigcup_{v\in V_{\lambda}}\set{C(Q,v)}$ which is relevant for the query.
In~\cite{LMN} searching for the predecessor is done using
binary search on the vertices of $Q$, which takes $O(\lg{\Delta})$, 
where $\Delta$ is the (hop) diameter of the graph. In general $\Delta$ may be as large as $n$.
Instead, one can record, for each vertex of $Q$, its ordinal number
along $Q$, and store the vertices of $\bigcup_{v\in
  V_{\lambda}}\set{C(Q,v)}$ in a fast integer predecessor data
structure (e.g.~\cite{Willard83}), which supports predecessor search
in $O(\lg\lg{\Delta})$ instead of $O(\lg\Delta)$. Therefore, the query
time for the vertex-label distance oracle of~\cite{LMN} can be made
$O(\eps^{-1} \lg n\lg\lg n)$ instead of $O(\eps^{-1} \lg^2 n)$ .

\section{Concluding Remarks}
This work presents an extension of Thorup's vertex-vertex oracle
to enable vertex-label queries in planar graphs.
Although our focus is on planar graphs, the algorithm works for any
class of graphs which are both minor-closed and tree-path
separable. By tree-path separable we mean that, given any spanning
tree, there exists a constant number of root branches whose deletion
separates the graph into subgraphs, each of size at most half the original size.

\paragraph{Related problems}
In this work we dealt with labeled-graphs where each
vertex has exactly one label. However, there is no obstacle
to deal with labeled-graphs in which each vertex has several labels.
Let $\kappa$ be the bound on the maximal number of labels a vertex is labeled by.
Assuming $\kappa = poly(n)$, the preprocessing time and space of the oracles in this work are multiplied by $\kappa$.
 
We note that the oracle in this paper can be used to obtain
vertex-vertex distance oracles in an apex graph $G$  (i.e., $G$ is
a planar graph with additional $\kappa$ apices).
Split each apex $a$ in $G$ into degree($a$) copies, one copy for each arc
incident to $a$. Label the copies of $a$ with a distinct label
$\lambda_a$. Note that the resulting planar graph $G'$ has the same number of
edges as $G$, and at most $\kappa$ times the number of vertices of $G$.
Construct a vertex-vertex distance oracle and a vertex-label distance oracle for $G'$. Also store
explicitly the distances between any two apices in $\kappa^2$ space. 
Consider a $u$-to-$v$ query.  Let $P$ be a shortest $u$-to-$v$ path in
$G$. If $P$ does not visit any apex, its length can be reported by the vertex-vertex distance oracle
for $G'$. If $P$ does visit some apex, then let $a_1,a_2$
be the first and last apices on $P$, respectively. The path $P$ can be
decomposed into ($i$) a  $u$-to-$a_1$  prefix, ($ii$) a $a_1$-to-$a_2$
infix, and ($iii$) a $a_2$-to-$v$ suffix. The lengths of the
$2\kappa$ possible prefixes ($i$) and suffixes ($iii$) can be found by making $2\kappa$ queries to 
the vertex-label distance oracle of $G'$. The apex-to-apex distances
($ii$) have been precomputed. Thus the query time is $O(\kappa^2 + \kappa(\lglg{n}\lglg(nN)+\eps^{-1}))$.

A possible direction for future work is to devise efficient label-to-label distance queries.
This seems  significantly more difficult. 
In the vertex-label query,
when queried for a $u$-to-$\lambda$ distance the algorithm used
the leafmost node $r$ in $\mathcal T$ that contains
$u$ and some $\lambda$-labeled vertex to quickly answer the query.
In a $\lambda_1$-to-$\lambda_2$ query, 
we do not know which node $r \in \mathcal{T}$ necessarily has its
separator intersected by the
desired shortest path.
This is because $\lambda_1$ and $\lambda_2$ labeled vertices might be scattered
in any of the leaves of $\mathcal{T}$.
Constructing  a label-to-label distance oracle whose query
 time does not depend on the number of vertices with the queried
 labels remains as an interesting open problem.
 

\bibliographystyle{spmpsci}
\bibliography{dist}



\end{document}